\documentclass[letterpaper, 11pt]{article}
\usepackage[margin=1in]{geometry}
\usepackage{latexsym}
\usepackage{amssymb}
\usepackage{xcolor}
\usepackage{color}
\usepackage{amsmath,amsfonts,amsthm}
\usepackage{graphicx,float}
\usepackage{enumitem}
\usepackage[ruled,vlined,resetcount, algosection]{algorithm2e}
\usepackage{setspace}
\usepackage{todonotes}
\usepackage{mathtools}
\usepackage{bbm}

\usepackage[pagebackref]{hyperref}
\makeatletter
\def\namedlabel#1#2{\begingroup
    #2%
    \def\@currentlabel{#2}%
    \phantomsection\label{#1}\endgroup
}
\makeatother

\hypersetup{
bookmarks=true,
colorlinks=true,
linkcolor=blue,
urlcolor=blue,
citecolor=blue,
pdftex,
linktocpage=true, % makes the page number as hyperlink in table of content
hyperindex=true
}
\usepackage{cleveref}

\newcommand{\supp}{\mathsf{supp}}

\newcommand{\poly}{\mathsf{poly}}

\newcommand{\eps}{\varepsilon}
\newcommand{\ent}{\mathsf{H}}

\newcommand{\F}{\mathbb{F}}

\newcommand{\Enc}{\mathsf{Enc}}

\newcommand{\ED}{\mathsf{ED}}
\newcommand{\HD}{\mathsf{HD}}

\newcommand{\LCS}{\mathsf{LCS}}

\newcommand{\cost}{\mathsf{cost}} 
\newcommand{\obj}{\mathsf{obj}} 

\newcommand{\idk}{\mathbbm{1}}

\newtheorem{lemma}{Lemma}[section]
\newtheorem{theorem}[lemma]{Theorem}
\newtheorem{claim}[lemma]{Claim}
\newtheorem{construction}[lemma]{Construction}

\newtheorem{definition}[lemma]{Definition}
\newtheorem{corollary}[lemma]{Corollary}

\newtheorem{remark}[lemma]{Remark}

\renewcommand{\P}{\mathrm{Pr}}

\begin{document}

\allowdisplaybreaks

\title{Efficient Linear and Affine Codes for Correcting Insertions/Deletions}
%Edit Distance}
\author{
Kuan Cheng\thanks{Center on Frontiers of Computing Studies, Computer Science Department, Peking University. {\tt ckkcdh@pku.edu.cn.} Supported in part by a start-up funding of Peking University, a Simons Investigator Award (\#409864, David Zuckerman) and NSF Award CCF-1617713.  } 
\and Venkatesan Guruswami\thanks{Computer Science Department, Carnegie Mellon University. {\tt venkatg@cs.cmu.edu.} Research supported in part by NSF grant CCF-1814603.}  
\and Bernhard Haeupler\thanks{Computer Science Department, Carnegie Mellon University. {\tt haeupler@cs.cmu.edu.} Supported in part by NSF awards CCF-1814603, CCF-1910588, NSF CAREER award CCF-1750808 and a Sloan Research Fellowship.} 
\and Xin Li\thanks{Department of Computer Science,
Johns Hopkins University. {\tt lixints@cs.jhu.edu.} Supported by NSF Award CCF-1617713 and NSF CAREER Award CCF-1845349.} }

\date{}

\maketitle
\thispagestyle{empty}

%%%%%%%%%%%%%%%%%%%%%%%%%%%%%%%%%%

\begin{abstract}

    This paper studies \emph{linear} and \emph{affine} error-correcting codes for correcting synchronization errors such as insertions and deletions. We call such codes linear/affine insdel codes.
    
    \smallskip
Linear codes that can correct even a single deletion are limited to have information rate at most $1/2$ (achieved by the trivial 2-fold repetition code). Previously, it was (erroneously) reported that more generally no non-trivial linear codes correcting $k$ deletions exist, i.e., that the $(k+1)$-fold repetition codes and its rate of $1/(k+1)$ are basically optimal
for any $k$. We disprove this and show the existence of 
binary linear codes of length $n$ and rate just below $1/2$ capable of correcting $\Omega(n)$ insertions and deletions. 
This identifies rate $1/2$ as a sharp threshold for recovery from deletions for linear codes, and reopens the quest for a better understanding of the capabilities of linear codes for correcting insertions/deletions.

 \smallskip
 We prove novel outer bounds and existential inner bounds for the rate vs. (edit) distance trade-off of linear insdel codes. We complement our existential results with an efficient synchronization-string-based transformation that converts any asymptotically-good linear code for Hamming errors into an asymptotically-good linear code for insdel errors. Lastly, we show that the $\frac{1}{2}$-rate limitation does not hold for affine codes by giving an explicit affine code of rate $1-\epsilon$ which can efficiently correct a constant fraction of insdel errors.

%We give an explicit construction of efficiently decodable linear ECCs with constant rate which for any $\delta < 1$ explicit linear ECCs correct a $1 - \delta$ fraction of insertions and deletions. 

\iffalse
    This paper reopens and significantly advances the study of linear and affine error correcting codes (ECCs) for correcting synchronization errors such as insertions and deletions. 
    
    \smallskip
    
    Previously it was reported that no non-trivial linear ECCs correcting $k$ deletions exist, i.e., that the $k+1$-fold repetition codes and its rate of $1/(k+1)$ are optimal for any $k$. We disprove this by showing that random linear codes of length $n$ and rate just below $1/2$ can correct a linear number $k = \Theta(n)$ of insertions and deletions  - instead of just a single deletion. This widely reopens the quest for what the possibilities and limitations of linear ECCs for synchronization errors are and for what parameters explicit and efficient codes can be given. 
    
    \smallskip
    
    This paper gives a near complete picture about the existential rate/distance tradeoff for such codes. We also give explicit constructions of efficiently decodable linear ECCs with constant rate which for any $\delta < 1$ explicit linear ECCs correct a $1 - \delta$ fraction of insertions and deletions. Lastly, we study and give new results for affine ECCs correcting synchronization errors. Many vexing questions remain. 
\fi

\end{abstract}
\newpage
\tableofcontents
\thispagestyle{empty}

\newpage
\section{Introduction}
\vspace{-1ex}
Error-correcting codes resilient to synchronization errors such as insertions and deletions (insdel errors) have witnessed a lot of exciting recent progress. Efficiently constructible and decodable codes with strong, or even near-optimal, parameters have been discovered in several regimes.

For codes over the binary or other fixed constant-sized alphabets, these include (i)  codes to correct a small (constant) number of insdel errors~\cite{SimaB19,KZXK18,BGZ-journal}, (ii) High-rate codes to correct a small fraction of insdel errors~\cite{Haeupler19,KZXK18,SychnStringsInteractive,GW}, (iii) codes of positive rate over an alphabet size $q$ to correct a fraction of insdel errors approaching $1-1/(q+\sqrt{q})$~\cite{BGH17}, (iv) near-optimal codes for block edit errors including transpositions~\cite{cheng_et_al:LIPIcs:2019:10613},
(v) list-decodable codes over large alphabets able to correct more than $100\%$ of insertions and achieving an optimal rate independent of the fraction of insertions~\cite{HSS18ListInsdel}, and (vi) list-decodable codes for any alphabet of positive rate for the largest possible combination of insertions and deletions that is information-theoretically correctible~\cite{GHS-stoc20,hayashi2018list,wachter2017list,GW}. Many of these recent results build on (ideas from) synchronization strings~\cite{CHLSW18,HS17} introduced in \cite{haeupler2017synchronization} which elegantly reduced questions of tackling synchronization errors to the much better understood problem of correcting Hamming errors and erasures. This also led to codes over large alphabets with the optimal trade-off between rate and fraction of insdel errors corrected~\cite{haeupler2017synchronization} and a near-linear encoding and decoding complexity~\cite{HRS19LinearTimeInsdel}.

%Synchronization strings have since been used in other contexts....\todo{Mention interactive coding? Cite refs}

Interestingly, not a single one of the above-mentioned constructions provides a linear code. Indeed, to the best of our knowledge, no linear insdel code constructions, beyond the trivial repetition codes, have appeared in the literature. This stands in stark contrast to error-correcting codes for Hamming errors where most important codes both in theory and practice are linear codes, including Hamming, Reed-Solomon and Reed-Muller codes, algebraic-geometry codes, polar codes, Turbo codes, expander codes, and LDPC codes.

Linear codes have many advantages: 1.) They have nice and compact representations via either their generator matrix $G$ or their parity check matrix $H$. 2.) These compact representations also directly imply that every linear codes has a simple and efficient $O(n^2)$-time encoding procedure $\mathrm{Enc}(x) = xG$, as well as an efficient test $Hx = 0$ for whether a given string $x$ is a codeword. Faster encoding algorithms are often possible, and when $H$ is sparse,  even linear time decoders have been discovered. 3.) Linear codes allow for simpler analysis and allow powerful methods from linear algebra to apply. For example, many properties of a linear code can be related to its (Hamming) weight distribution, including the minimum distance of a linear code which is equal to the minimum weight of any non-zero codeword. 

One of the reasons why the recent wave of results on error correcting codes for synchronization errors did not produce any results on linear insdel codes are two negative results suggesting that no such non-trivial codes exist. In particular, in 2007 \cite{AGFC07} proved that any linear code that even corrects a single deletion must have a rate of at most $1/2$, which is the rate achieved by the trivial repetition code that outputs every symbol twice. This result is independent of the size of the alphabet and stands in stark contrast to codes for Hamming errors or non-linear insdel codes which can correct a constant \emph{fraction} of errors while having a rate arbitrarily close to $1$. In \cite{BGZ-journal} it was claimed that the impossibility result of \cite{AGFC07} extends to more than one deletion. In particular, the rate achievable by any linear code correcting a fixed number $k$ of deletions is upper bound by (about) $1/(k+1)$, which is the rate of the trivial $(k+1)$-fold repetition code. 

\medskip
The starting point for this paper was the realization that, while the impossibility result of \cite{AGFC07} stands, the proof of its extension to larger $k$ in \cite{BGZ-journal} is flawed. It has since been retracted in \cite{BGZ-arxiv}, as a result of this work. This reopens the quest to better understand the potential and limitations of linear codes in tackling insertion and deletion errors.

In this paper we provide new existential results as well as explicit constructions of efficient linear codes. Our results are the first to provide any non-trivial linear insdel codes that drastically outperform simple repetition codes. 
We remark that throughout this paper we consider worst-case error corretion (or called zero-error channel coding in some other literatures), i.e. as long as the number of errors is below a certain threshold, the error correcting procedure can always recover the correct message.
This is stronger than the probabilistic insertion/deletion correction setting, e.g. \cite{hanna2018guess}, which corrects errors that are from a  distribution  independent of the message, succeeding with a high probability.

\subsection{Combinatorial/existence results}
We start by showing that quite powerful linear insdel codes exist just below the $1/2$ rate upper bound of \cite{AGFC07}. In particular, a random binary linear code with a rate just below $1/2$ is capable of correcting a linear number of errors, vastly more than the single error of the $2$-fold repetition code of rate $1/2$.

\begin{theorem}\label{simplerandombinary}%[Existential linear codes]
For any $\eps  > 0$, a random binary linear code of length $n$ and rate $1/2 - \eps$ has edit distance\footnote{Edit distance between two strings is the minimum number of insertions, deletions and replacements that can modify one string to be the other.} 
$\Omega(\eps \log^{-1} \frac{1}{\eps}) n$ with   probability at least $1-2^{-\Theta(n)}$.
\end{theorem}

We also extend the result of \cite{AGFC07} into a distance-dependent rate upper bound showing that, independent of the alphabet used, no rate larger than $\frac{1}{2}(1-\delta)$ is achievabe by any linear code which can correct a $\delta$-fraction of insdel errors. This shows that the asymptotics of Theorem~\ref{simplerandombinary} are optimal up to the $\Omega(\log^{-1} \frac{1}{\eps})$ factor. It also shows that while linear codes for Hamming errors and non-linear insdel codes can get arbitrarily close to the Singleton bound, the best rate achievable by any linear insdel code (over a finite alphabet) is exactly half as large, for any error correcting capability $\delta >0$:

\begin{theorem}[Half-Singleton bound]\label{halfsingletonboundupperandlower}
For any $\delta >0$ and any $\eps  >0$, for all large enough finite fields of size $q =2^{ \Theta(\eps^{-1}) }$, there exists a family of $\F_q$-linear insdel codes of rate $\frac{1}{2}(1 - \delta) - \eps $ capable of correcting any $\delta$-fraction of insdel errors. Furthermore, no such family of linear codes exists with a rate larger than $\frac{1}{2}(1 - \delta)$ over any alphabet.
\end{theorem}

We also prove an alphabet-dependent rate upper bound, namely a ``half-Plotkin" bound. This bound, which is rather subtle to establish, shows that the rate of a $\F_q$-linear code capable of correcting a fracton $\delta$ of insdel errors can be at most half as large as the $(1 - \frac{q}{q-1}\delta)$ Plotkin outer bound for $q$-ary codes of relative (Hamming) distance $\delta$. 
%correcting any $\delta$ fraction of Hamming errors. 

\begin{theorem}[Half-Plotkin outer bound]\label{halfplotkin}
Any $\F_q$-linear code which can correct a $\delta$ fraction of deletions has rate at most $\frac{1}{2}(1 - \frac{q}{q-1}\delta) + o(1)$.
\end{theorem}
\subsection{Explicit constructions of linear/affine insdel codes}
We also give an explicit construction of asymptotically good linear insdel codes. Using the ideas from synchronization strings~\cite{haeupler2017synchronization} we give a transformation which takes any explicit, efficiently decodable linear code $C$ for Hamming errors and produces an explicit, efficiently decodable linear insdel code $C'$ whose rate and error correcting capabilities are at most a constant factor smaller than those of the code $C$.

\begin{theorem}
\label{thm:linear-const-intro}
Fix any finite field $\F_q$.
There is an explicit construction of a $\F_q$-linear code family with rate bounded away from zero, a linear time encoding algorithm, and a polynomial time decoding algorithm to correct a positive constant fraction of  insertions and deletions.  The generator matrix of codes in the family can be deterministically computed in polynomial time in the code block length.
%\todo{Is the decoder quadratic time or $O(n^4)$ time? I just say polytime here}
\end{theorem}

The above result gives an analog of the classic result of Schulman and Zuckerman which gave an explicit construction of a \emph{non-linear} binary code family with positive rate and an efficient algorithm to correct a constant fraction of insdel errors~\cite{SZ99}. The construction in \cite{SZ99} starts with an outer code that can correct a constant fraction of Hamming errors (like Reed-Solomon codes), adds index information to its symbols, and then encodes the resulting indexed symbols by a suitable inner codes. The introduction of index information precludes getting a linear code via this approach.

Our linear code construction of Theorem~\ref{thm:linear-const-intro} does not have rate approaching $1/2$, which is shown to be possible by the existence result of Theorem~\ref{simplerandombinary}. Improving the rate of our construction remains a fascinating open question. An ultimate goal would be to construct linear codes approaching the half-Singleton bound over large finite fields.

A systematic code encodes a message $x$ as $x$ followed by some check symbols. For Hamming errors, one can make any linear code systematic by performing a basis change in the generator matrix and permuting the symbols if necessary. For insdel errors, permuting symbols can greatly alter the edit distance. We show our random coding result underlying Theorems \ref{simplerandombinary} and \ref{halfsingletonboundupperandlower} holds with the same parameters also for systematic linear codes. We can also get a version of Theorem~\ref{thm:linear-const-intro} with systematic codes by simply appending the message at the beginning and losing constant factors in the rate and fraction of insdel errors corrected.

%Whether once can obtain explicit and efficient linear codes which achieve the Half-Singleton-Bound established to be the existentially tight rate-distance tradeoff for linear insdel codes by Theorem~\ref{halfsingletonboundupperandlower} remains a vexing open question. 
\vspace{-1ex}
\subsection{Affine insdel codes}
Our final result concerns affine codes, a concept introduced in \cite{AGFC07} to overcome the limitations of linear and cyclic codes for deletion errors. Indeed, it was suggested in \cite{AGFC07} to add a fixed symbol to the codewords of a cyclic/linear code to increase its edit distance. The codewords of such an encoding form an affine subspace. In general, an affine code is one with encoding function $\text{Enc}: \F_q^m \rightarrow \F_q^n$ given by an affine map $\text{Enc}(x) = x G + b$ for some generator matrix $G$ and offset vector $b$. (In fact, the affine codes constructed in this paper are of the same simpler form as the codes in \cite{AGFC07} where each position of a codeword is either a fixed symbol or a linear combination of input symbols.) While it is easy to see that affine codes and linear codes are equivalent in power for correcting Hamming errors, the same is not true for insdel errors. Indeed, despite being structurally about as simple as linear codes the impossibility result of \cite{AGFC07} and our outer bounds for linear insdel codes do \emph{not} apply to affine codes. This opens up the possibility of obtaining affine codes with better rates. 

%beyond the $1/2$ impossibliitythe Half-Singleton-Bound of Theorem~\ref{halfsingletonboundupperandlower}.

We show that this is indeed the case by giving high-rate affine insdel codes to correct a constant fraction of insdel errors while achieving a rate arbitrarily close to $1$. Our affine codes are binary, explicit, and feature efficient encoding and decoding algorithms. The construction is based on the binary codes in \cite{SychnStringsInteractive} which are a synchronization string based simplification of the codes of \cite{GW}. 

\begin{theorem}
\label{Thm:explicitAffineCode}
For any $\eps  >0$ there exists an explicit binary affine code family with rate $1 - \eps $ that can be efficiently decoded from a fraction $O(\eps^3)$ of insertions and deletions.
\end{theorem}
We note that unlike linear codes, affine codes can allow the introduction of index information as in the original Schulman-Zuckerman construction~\cite{SZ99}. However, their construction also requires inner codes with strong properties (such as a good density of $1$'s in every codeword, or good edit distance between all pairs of long enough subsequences of distinct codewords) which are difficult to ensure with affine codes, and in any case incompatible with obtaining high rate.

For general insdel codes, constructions with a better trade-off between rate and fraction of correctable insdel errors than the guarantee of Theorem~\ref{Thm:explicitAffineCode} are known. The constructions in \cite{GW,SychnStringsInteractive} could correct a fraction $\widetilde{\Omega}(\eps^2)$ of insdel errors with rate $1-\eps$. Even more recently, via connections to the document exchange problem, codes with rate $1-\eps$ that can correct a fraction $\Omega(\eps/\log^2 (1/\eps))$ of insertions and deletions have been constructed in \cite{KZXK18,Haeupler19}. 
%\todo{Are these the correct citations?}

\medskip
Our results expose the under-appreciated power of linear and affine codes in the context of correcting insdel errors and identify some limitations that are in contrast with general codes. Put together, they substantially improve our understanding of the landscape of possibilities for using linear/affine codes to tackle synchronization errors, and open up several intriguing questions that merit further study (see Section~\ref{sec:conclusion}).

\medskip\noindent\textbf{Outline.} The paper is organized as follows. We give a high level overview of our techniques, particularly those underlying our constructions of linear and affine codes, in Section~\ref{sec:overview}. The existential random coding bounds for linear codes, including their systematic versions, are presented in Section~\ref{sec:bound}. Our outer half-Singleton and half-Plotkin bounds are established in Section~\ref{sec:outer-bound}.  In Section~\ref{sec:monte-carlo} we present a Monte Carlo construction of linear insdel codes which is then derandomized in Section~\ref{sec:derand} to give the explicit codes promised in Theorem~\ref{thm:linear-const-intro}. The high-rate affine codes guaranteed by Theorem~\ref{Thm:explicitAffineCode} are presented in Section~\ref{sec:affine}. After a brief discussion about the construction of systematic linear insdel codes in Section~\ref{sec:systematic-linear}, we conclude the paper with a number of interesting open questions in Section~\ref{sec:conclusion}. Some of the technical proofs are deferred to Appendix~\ref{app:skipped-proofs}.

%Overall our results very much map out the possibilities and limitations of linear and affine codes for synchronization errors. $\ldots$ GIVE A STRONG ENDING HERE.

%This confirms that that affine codes have the potential to be as powerful  

%In fact, given that affine codes retain many of the advantages of linear codes but avoid the impossiblity results given by our Half-Singleton-Bound for linear codes, affine codes have the potential to be an interesting sweet spot combining the simplicity of linear codes.

%We show 

%While for simplicity and concreteness we focus on binary codes, we believe, but have not checked, that our approach extends to linear codes over any fixed finite field $\F_q$.

%
%
%\begin{theorem}
%
%There exists an explicit family of    asymptotically good  linear   codes for edit distance, where the information rate is $ O( \delta^{-\frac{1}{2}} (1-\Theta(\sqrt{\delta} \log \delta^{-1}   ))  )$ with $\delta$ being the error rate.
% 
%\end{theorem}

\vspace{-1ex}\section{Overview of our approach}
\label{sec:overview}
\vspace{-1ex}
Our existential results follow from relatively straightforward applications of the probabilistic method: we take a random generator matrix and analyze its property. The half-Singleton bound also follows easily from a reduction to previous results: if the code can correct $\delta n$ insdel errors then we can simply remove the first $\delta n-1$ symbols and the code can still correct one deletion, which by the result of \cite{AGFC07} limits the rate to $1/2$. The proof of the half-Plotkin bound is significantly more involved. From a $q$-ary linear code $C$ that can correct $\delta n$ insdel errors, we drop a prefix consisting of the first $d \approx q n \delta/(q-1)$ symbols of codewords and prove that the remaining punctured code $C'$ has dimension at most $\approx \tfrac{n}{2} (1-q \delta/(q-1))$. For the latter, we give a clever reduction to the fact that in any linear code of large enough dimension, there is a non-zero codeword which has many $0$'s in \emph{both} the first and the second half. We establish this latter fact using a carefully designed probabilistic argument combined with collapsing together linearly dependent coordinates. Now note that if instead of dropping the prefix consisting of the first $d \approx q n \delta/(q-1)$ symbols of all codewords, we just delete all symbols except $d/q$ of the most frequently occurring element in the prefix, then after this operation any two different codewords in $C$ are still distinct (since $d(1-1/q) \leq \delta n$). Note that the resulted codeword is of the form $\alpha^{d/q}$ followed by a codeword in $C'$, for some $\alpha \in \F_q$. This implies that the number of codewords in $C$ is at most $q$ times the number of codewords in $C'$, thus the original code has dimension at most $\dim(C')+1$ which is also at most $\approx \tfrac{n}{2} (1-q \delta/(q-1))$.

Now we give an informal overview of our constructions of linear codes and affine codes for insertions and deletions. As observed in many previous works, the main difficulty in handling edit errors is that insertions and deletions can shift the string, and thus cause the loss of index information. An elegant idea to handle this, introduced in the work of Haeupler and Shahrasbi \cite{haeupler2017synchronization}, and later used in \cite{HS17c, CHLSW18} is to construct a fixed string that can be used to recover the index information. Such a string is called a \emph{synchronization string} and one can then use it to transform a code for Hamming errors into a code for edit errors, just by appending the symbols of the synchronization string to every codeword. 

However, even if we start with a linear code for Hamming errors, the resulting code after appending the synchronization string is not necessarily a linear code or an affine code. This is due to both the fact that the synchronization string is some fixed string, and the fact that when appending we are putting the symbol from the codeword and the symbol from the synchronization string together to form a symbol in a larger alphabet. If we restrict ourselves to binary codes, then the above approach has an additional problem of increasing the alphabet size. Instead, here we use similar ideas, but modify the synchronization string to obtain a linear or affine code without blowing up the alphabet size. Specifically, instead of appending a synchronization string, we insert a fixed sequence of synchronization symbols into the codewords of a code for Hamming errors. This ensures that the resulting code is still a linear or affine code. We call such a sequence of symbols a \emph{synchronization separator sequence}, and we further show how to explicitly design such sequences to allow proper decoding. We now give more details below.%that by carefully choosing the number of $0$'s inserted, the sequence of $0$'s behaves like a synchornization string. In particular, we first show that by randomly choosing the number of $0$'s from some range, the resulted sequence has the desired properties, we then show how to derandomize this and achieve an explicit construction. 

\subsection{Randomized construction of the linear code}
We  start with a randomized construction and then describe how to derandomize it to get an explicit code.

To encode a given message $x \in \mathbb{F}_q^m$, we first encode it to be $y = C(x)$ using a  linear code $C:   \mathbb{F}_q^m \rightarrow  \mathbb{F}_q^{n_C}$ which can correct $\kappa_C$ Hamming errors. 
Then consider the following transformaion.
Before each symbol of $y$, we insert a short string of $0$'s, obtaining the  codeword
$$z = S[1] \circ y[1] \circ S[2] \circ y[2] \circ \cdots \circ S[n_C] \circ y[ n_C ].$$
Here $\circ$ denotes the string concatenation operation. $S = (S[1], S[2], \ldots, S[n_C])$ is a sequence of independent random variables s.t. each $S[i]$ is an all-$0$ string with length being uniformly chosen from $\{ 1,2,\ldots, a \}$ for a properly chosen parameter $a$. 
Denote this new code as $C_S$.
Notice that in order to have constant rate, we only need to let $C$ have constant rate and set $a$ as a constant.
Also notice that only inserting $0$'s can indeed keep the linearity of the code.
But the question is why this can give a code for edit distance.

To see that $C_S$ has large edit distance with high probability, we directly show our decoding.
The decoding will first compute some kind of optimal matching between the received word and a template of codeword to recover $y = C(x)$, and then use the decoding algorithm of $C$ to recover $x$.

Specifically,  let the received corrupted codeword be $z'$, assuming there are $\le \kappa = 0.01 \kappa_C  $ insdel errors.
We imagine that the correct codeword is in the following form 
$$z_? =  S[1] \circ ? \circ S[2] \circ ? \circ \cdots \circ S[n_C] \circ ? ,$$
where each ``$?$'' mark is a special symbol indicating a blank that should be filled by a symbol of $y$. To fill the blanks, we
want to compute a monotone matching between $z'$ and $z_?$, matching the ``$?$'' marks of $z_?$ to the non-zero-symbols of  $z'$. 
Further the matching should have the number of non-zero symbols in $z'$ that are matched to their correct   positions to be large enough.
If the decoding can always provide such a matching, then  we can fill each matched ``$?$'' mark  with the symbol it is matched to and each  unmatched ``$?$'' mark  with a $0$.
Notice that
this will give a string $y'$ which is close enough to $y$ and thus can be decoded correctly. 

Now we only need to show the matching process. 
In a first thought, it seems hard to give such a process, since   the original positions  for   symbols are not known and hence one does not know how to establish the matching.
However we can  use  $S$ to settle this issue in a subtle way.
We define an objective function $\obj$ and   the matching process is just finding a matching $w$ that maximizes $\obj(w)$.
For a  matching $w $,  each match $(i, j) \in w$  matches   the $i$-th question mark of $z_?$ to the $j$-th non-zero of $z'$. 
Define the cost function $ \mathsf{cost}(w) $ to be 
$$ \cost(w) =  \sum_{(i, j)\in w} \mathbbm{1}\left( p_i - p_{i'} \neq q_j - q_{j'} \right) .$$
Here   $(i' , j')$ is the immediate previous match of $(i, j)$. 
$p_i$ is the position of the $i$-th question mark in $z_?$.
$q_j$ is the position of the $j$-th non-zero symbol in $z'$.
$\mathbbm{1}() $ is the indicator function. 
There are corner cases where there are no immediate previous match. 
We don't handle them here but will handle them in the main body of the paper.
Define $\obj(w)=|w| -  \cost(w)$.

We only need to show why this  ensures a small Hamming distance between $y'$ and $y$, as long as the number of insertions and deletions is $\kappa\le 0.01 \kappa_C$.
We will  prove the contrapositive, that is, if the Hamming distance between $y'$ and $y$ is larger than $\kappa_C$, then this will contradicts the optimality of $w$ returned by the matching process. 
%For simplicity let us assume that we are working with a binary alphabet.
%This is due to the following reasoning.

First we bound the number of non-zero symbols  in $z'$ that are not matched, to be $\le 3\kappa $, as in Lemma \ref{lem:notmatchedones}. This follows from a careful analysis on how the insdel errors affect  the matching on its objective function.

Next we focus on the matched symbols of $z'$ and analyze how many of them are not matched to their correct positions in $y$. Towards that, we say a match  $(i, j)$ in $w$ is bad if the  $j$-th non-zero symbol of $z'$ is not inserted but $i$ does not correspond to its original position. 

Now suppose the Hamming distance between $y$ and $y'$ is larger than $\kappa_C$, and let $\Delta(w)$ be the number of bad matches in $w$.\ We claim that $\Delta(w)>(\kappa_C - 6\kappa)/2$, i.e., Lemma \ref{lem:badmatchcount}.  
To see this, recall that by Lemma \ref{lem:notmatchedones}, there are at most $O(\kappa)$ non-zeros in $z'$ that are not matched. So a large fraction of the non-zero symbols in $z'$ are matched. Intuitively, since the Hamming distance between $y$ and $y'$ is larger than $\kappa_C \gg \kappa$, a certain fraction of the  matches must be bad. We show that this is indeed true, although some effort and careful analysis are required here, since we need to deal with several issues such as the non-zeros in $z'$ may be inserted, and the Hamming errors may come from the unmatched ``$?$'' marks.

Since our definition of a bad match implies that the matched non-zero symbol in $z'$ is not inserted, the above bad matches can be viewed as inducing a self matching $\tilde{w}$ between the ``$?$'' marks of $z_?$ and itself. Now we   use a crucial observation that if a match $(i, j)$ is bad, then the lengths of the corresponding two intervals involving this match are unlikely to be equal, i.e. $ p_i - p_{i'} \neq q_j - q_{j'}$ with high probability. Indeed, for any bad match $(i, j)$, we can show that it will contribute $1$ to $\cost$ with probability $1-1/a$, since the length of the $0$'s is uniformly chosen from $[a]$. Given that there are $\Delta(w)$ bad matches, with high probability they will contribute $0.1 \Delta(w) $ to $\cost$. Now note that $\kappa$ errors can decrease $\cost$ by at most $\kappa$, since each insertion or deletion can only affect at most one of the indicators in $\cost$. As a result, we have $\cost(w) \geq 0.1 \Delta(w) -\kappa$ and thus $\obj(w)=|w|-\cost(w) \leq \tau+\kappa-(0.1\Delta(w) -\kappa)=\tau+2 \kappa- 0.1\Delta(w) $, since we set $\kappa =0.01\kappa_C$ and $\Delta(w)>(\kappa_C - 6\kappa)/2$. But  this contradicts the fact that $\obj(w) \geq \obj(w^*) \geq \tau - 2\kappa  $.

\subsection{Synchronization separator sequence and the explicit linear code}
From our discussion above, we can actually abstract out the property of the sequence $S$ that we need. In short, we need the following property which holds with high probability from a random sequence: any self matching between the $?$ marks which has many bad matches must incur a large cost. Thus, to derandomize our construction, we just need to explicitly construct such a sequence $S$. Towards that, we give the following definition.

For a sequence $ s[1] \circ ? \circ s[2] \circ ? \circ \cdots \circ s[n ] \circ ? $, a  $?$-to-$?$ self matching $w$ is a monotone matching between the string and itself s.t. each match $(i, j)$ matches the $i$-th question mark to the $j$-th question mark. A match  $(i, j)$ in $w$ is called \emph{undesired} if 
\begin{itemize}
  
\item $i\neq j$ (i.e. the match is bad);
\item $   p_i-p_{i'} =  p_j - p_{j'} $ for any $(i,j)$ that has an immediate previous match $(i', j')$.
\end{itemize}
Recall that $p_{i }$ is the position of the  $i $-th $?$-mark in $z_?$.

\begin{definition}[$(\Lambda, a)$ synchronization separator sequence]

$s[1], s[2], \ldots, s[n]$ is called a $(\Lambda, a)$ synchronization separator sequence, if for any self matching of $ s[1] \circ ? \circ s[2] \circ ? \circ \cdots \circ s[n ] \circ ?   $, the number of undesired matches is at most $\Lambda$.

Here each $s[i], i\in [n]$ is an all $0$ string whose length is in $\{ 1,\ldots,a \}$.
\end{definition}

It can be easily seen that a $(\Lambda, a)$ synchronization separator sequence satisfies the property we require: any matching with $\Delta$ bad matches must incur a cost of at least $\Delta-\Lambda$. In this paper, for any $n\in \mathbb{N}, \Lambda \leq n $, we can construct an explicit synchronization separator sequence of length $n$, with $ a =(\frac{n}{\Lambda})^{ \Theta(1)}$. In particular, setting $\Lambda=\alpha n$ for some constant $0< \alpha<1$ results in $a=O(1)$ and an explicit asymptotically good linear code for insertions and deletions.

We now show how to construct this object explicitly. The high level idea is that we use a pseudorandom distribution, i.e., an $\eps$ almost $k $-wise independent distribution with small seed length, to generate the sequence $S$. To show this works, we reduce the condition of being a $(\Lambda, a)$ synchronization separator sequence to some small scale checks, i.e., for every pair of short substrings (the sum of the lengths of the two substrings is $\Theta\left( \frac{n}{\Lambda} \frac{\log n }{ \log \frac{n}{\Lambda} } \right)$), we check if there is no matching between them that has a large fraction of undesired matches. Note that these checks are local checks, so they can be fooled by an almost $k $-wise independent distribution with a proper setting of the parameter $k$. That is, under such a pseudorandom distribution, the probability that all checks passed is positive. Furthermore this distribution can be generated by a small number of random bits (i.e., $\Theta(\log n)$ random bits), and thus we can exhaustively search the entire sample space to find an explicit sequence that passes all checks. Note that for a specific sequence, whether it is a $(\Lambda, a)$ synchronization separator sequence can be checked efficiently (e.g., by dynamic programming), so the whole construction runs in polynomial time.
%As the support for the pseudodistribution is small, we can exaustively search out a good sample point. 
We emphasize that to show that passing all small scale checks guarantees a good synchronization separator sequence, we use a matching cutting argument  similar to that of \cite{KZXK18}. But the situation is slightly different, so there are indeed some efforts and careful analysis required.

\smallskip\noindent \textbf{Systematic linear code.}
Our construction can be turned into a systematic linear code by simply concatenating the message with the corresponding codeword. To decode, we can simply ignore the first $m$ symbols ($m$ is the message length) and decode the rest of the received word. See Section~\ref{sec:systematic-linear} for the details.
 
\subsection{The high rate affine code construction}
In order to construct an affine code $C$ over the binary alphabet, our first step is to use the construction in \cite{haeupler2017synchronization} which works for a larger alphabet. Specifically, we take a linear code $C_0$ over the alphabet $\mathbb{F}_2^{l_0}$ for Hamming errors, and for each codeword we concatenate it coordinate-wise with a synchronization string over a constant size alphabet. This provides us with a code $C'$ for insertions and deletions. However, due to the coordinate-wise concatenation, the resulted code is not linear or affine any more. 

We can convert the code $C'$ back into an affine code by directly expressing each symbol in binary. Since we start with a linear code $C_0$ over the alphabet $\mathbb{F}_2^{l_0}$, the conversion of $C_0$ into the binary alphabet keeps its linearity. The symbols of the synchronization string now become fixed strings, and hence now the code becomes affine. However, simply doing this will be problematic, since now in a corrupted codeword it will be hard to tell the boundaries between different symbols. 

To solve this problem, we create special boundaries in the codewords by taking a special string $B$ which consists of $t+1$ $1$'s, for some parameter $t$. We
insert the string $B$ before each symbol in a codeword $y'$ of $C'$. To distinguish the boundaries from the bits in a codeword, we also modify the codeword in the following way: when we express each symbol of $y'$ in binary, we insert a $0$ after every $t$ bits.
We refer to the binary representation of a symbol in $y'$, and the inserted $0$'s as the content. We refer to the content together with the boundary before it as a block. This finishes our construction.

The decoding strategy and correctness are as follows.
Assume there are $\kappa$ insertions and deletions. Note that
each insertion or deletion can only affect at most two blocks, since the worst situation is that it creates a new boundary and also modifies the content of a block.
Therefore there are only $O(\kappa)$ corrupted blocks. Now we can just locate every correct boundary and view the substrings between two adjacent boundaries as contents. By doing this, every uncorrupted block of $y'$ will be recovered correctly. Now we can use the decoding algorithm in \cite{haeupler2017synchronization} to recover the message, as long as there are at most $\kappa$ insertions/deletions. 

To get a high information rate, we start with a linear code $C_0$ that has codeword length $n_0$, decoding radius $\kappa_0 = \eps n_0$ and information rate $1-\Theta(\eps)$. This code has  $l_0 = O(\eps^{-2})$, and can be constructed using algebraic geometry codes. We set $t = O( \frac{1}{\eps})$. Note that the synchronization string in \cite{haeupler2017synchronization} has constant alphabet size, thus it follows from our construction that the new code $C$ has codeword length $O(l_0n_0+t n_0)=O(\eps^{-2}n_0)$. From our previous discussion, the number of insertions and deletions $C$ can correct is $\kappa=\Theta(\kappa_0)=\Theta(\eps^{3} n)$. For the information rate, notice that $l_0 $ is much larger than $t$, and the synchronization string has constant alphabet size, thus the information rate of $C$ is still $ 1-\Theta(\eps) $.

Since the only modification we make to the codewords is inserting some fixed strings in fixed positions, the code is still an affine code.

\section{Preliminaries}
\textbf{Notation.} Let $\Sigma$ be an alphabet. For a string $x\in \Sigma^*$,
\begin{enumerate}
\item $|x|$ denotes the length of the string.
\item $x[i,j]$ denotes the substring of $x$ from position $i$ to position $j$ (both endpoints included).
\item $x[i]$ denotes the $i$-th symbol of $x$.

\item $x\circ x'$ denotes the concatenation of $x$ and some other string $x'\in \Sigma^*$.
%\item $B$-prefix denotes the first $B$ symbols of $x$. (Usually used when $\Sigma = \{0, 1\}$.)
\item $x^N$ the concatenation of $N$ copies of the string $x$.
\end{enumerate}
\subsection{Edit distance and longest common subsequence}
\begin{definition}[Edit distance] For any two strings $x, y\in\Sigma^n$, the edit distance $\ED(x, y)$ is the minimum number of edit operations (insertions and deletions) required to transform $x$ into $y$.\footnote{The standard definition of edit distance also allows substitution, but for simplicity we only consider insertions and deletions here, as a substitution can be replaced by a deletion followed by an insertion.}
\end{definition}
\begin{definition}[Longest Common Subsequence] For any two strings $x, y$ over an alphabet $\Sigma$, a longest common subsequence of $x$ and $y$ is a longest pair of subsequences of $x$ and $y$ that are equal as strings. We use $\LCS(x, y)$ to denote the length of a longest common subsequence between $x$ and $y$.
\end{definition}

Note that $\ED(x, y) = |x|+|y|-2\cdot \LCS(x, y)$.

\subsection{Error correcting codes}
\begin{definition}
An $(n, m ,d)$-code $C$ is an error-correcting code (for Hamming errors) with codeword length $n$, message length $m$, such that the Hamming distance between every pair of codewords in $C$ is at least $d$.
\end{definition}
Next we recall the definition of error-correcting codes (ECC) for edit errors.
\begin{definition}
Fix an alphabet $\Sigma$, an error-correcting code $C\subseteq \Sigma^n$ for edit errors with message length $m$ and codeword length $n$ consists of an encoding function $\text{Enc}:\Sigma^m\rightarrow \Sigma^n$ and a decoding function $\text{Dec}:\Sigma^*\rightarrow \Sigma^m$. The code can correct $k$ edit errors if for every $y$, s. t. $\ED(y,\text{Enc}(x))\leq k$, we have $\text{Dec}(y) = x$. The rate of the code is defined as $\frac{m}{n}$. The alphabet size is $|\Sigma|$.
\end{definition}

The code family $C$ is explicit (or has an explicit construction) if both encoding and decoding can be done in polynomial time. We say $C$ is a linear code if the alphabet $\Sigma$ is a finite field $\F_q$ and the encoding function $\text{Enc}:\F_q^m\rightarrow \F_q^n$ is a $\F_q$-linear map.

Let $C_1, C_2$ be two linear codes over $\F_q^{n}$.
We use $C_1 \cap C_2$ to denote the intersection of the two linear spaces, which is also a linear space.
Also we use $C_1 \cup C_2$ to denote the linear space which is spanned by all basis vectors of $C_1, C_2$.

%\subsection{Almost k-wise independence}
\subsection{Pseudorandom generator}
We use $U_n$ to denote the uniform distribution on $\{0,1\}^n$, and $\idk(\cdot) $ to denote the indicator function.

\begin{definition}[$\eps$-almost $k$-wise independence \cite{alon1992simple}]
Random variables $X_1, X_2, \ldots, X_n $  over $\{0,1\}$ are $\eps$-almost $k$-wise independent in max norm if for every distinct $ i_1, i_2, $ $\ldots, i_{k} \in [n]$, $\forall x \in \{0,1\}^{k}$, $|\Pr[ X_{i_1} \circ X_{i_2} \circ \cdots \circ X_{i_{k}} = x] - 2^{-k} | \leq \eps.$

A function $g: \{0,1\}^{d} \rightarrow \{0,1\}^{n}$ is an $\eps$-almost $k$-wise independence generator in max norm if $g(U) = X = X_{1} \circ \cdots X_{n}$ are $\eps$-almost $k$-wise independent in max norm.
\end{definition}
In the following passage, unless specified, when we say $\eps$-almost $k$-wise independence, we mean in max norm.

\begin{theorem}[$\eps$-almost $k$-wise independence generator \cite{alon1992simple}]
\label{almostkwiseg}
There exists an explicit construction s.t. for every $n, k \in \mathbb{N}$, $\eps > 0$, it computes an $\eps$-almost $k$-wise independence generator $g: \{0,1\}^{d} \rightarrow \{0,1\}^n$, where $d = O(\log \frac{k \log n }{\eps})$.

The construction is highly explicit in the sense that, $\forall i\in [n]$, the $i$-th output bit can be computed in time $\poly(k, \log n, \frac{1}{\eps})$ given the seed and $i$.
\end{theorem}

\section{Existential Bounds of Linear Codes}\label{sec:bound}
In this section we show existential bounds of linear codes for insdel errors. We have the following theorem.

\begin{theorem}
\label{thm:existence-linear}
Let $\F_q$ denote the finite field of size $q$. For any $\delta>0$ there exists a linear code family over $\F_q$  that can correct up to $\delta n$ insertions and deletions, with rate $(1-\delta)/2-\ent(\delta) / \log_2 q$. Here $\ent(\delta)=-\delta \log_2 \delta-(1-\delta) \log_2(1-\delta)$ is the binary entropy function.
\end{theorem}
We remark that if $q$ is a constant, then both $\delta$ and the rate can be constants i.e. the code is asymptotically good.
\begin{proof}
We show for all large enough $n$ there is such a code.
Let $m$ be an integer so that $m/n \le (1-\delta)/2 - \ent(\delta)/\log_2 q$.
Consider the random linear code with encoding function $y =x G$, where $x \in \F^m_q$ is the message, $y \in \F^n_q$ is the codeword, and $G \in \F^{m \times n}_q$ is an $m$ by $n$ matrix where each entry is chosen independently and uniformly from $\F_q$.

To ensure that the code can correct up to $\delta n$ insertions and deletions, we just need to make sure that every two different codewords have edit distance larger than $2\delta n$, or equivalently the length of the longest common subsequence is less than $(1-\delta)n$. We now bound the probability that this does not happen. First we have the following claim the proof of which is deferred to Appendix~
\ref{app:skipped-proofs}. 

\begin{claim}
\label{claim:pairwise-ind}
For any two different messages $x, x'$ and codewords $C=x G, C'=x' G$, consider an arbitrary length $t$ subsequence, with indices  $\{s_1, \cdots, s_t\}$, of $C$ and an arbitrary length $t$ subsequence,  with indices $\{r_1, \cdots, r_t\}$, of $C'$. Then 
\[\Pr[\forall k \in [t], C_{s_k}=C'_{r_k}] \leq q^{-t}.\]
\end{claim}

Using the claim, and noticing that the total number of possible cases where two strings of length $n$ have a common subsequence of length $(1-\delta)n$ is at most $\binom{n}{(1-\delta)n}^2$, we have
\begin{align*}
\Pr[\LCS(C, C') \geq (1-\delta)n] &\leq  \binom{n}{(1-\delta)n}^2 q^{-((1-\delta)n)} \\
&=  \binom{n}{\delta n} ^2 q^{-((1-\delta)n)} \leq 2^{2\ent(\delta)n} q^{(\delta-1)n}.
\end{align*}
Thus by a union bound
\[\Pr[\exists C \neq C', \LCS(C, C') \geq (1-\delta)n] < q^{2m} 2^{2\ent(\delta)n} q^{(\delta-1)n}.\]
The right hand side is at most $1$ as long as $m/n \le (1-\delta)/2 - \ent(\delta)/\log_2 q$. Thus the rate of the code family can be made to approach $(1-\delta)/2 - \ent(\delta)/\log_2 q$.
\end{proof}

In particular, for the case of binary and large finite alphabets we have the following corollaries.

\begin{corollary} 
\label{cor:binary-existence}
For any $\delta>0$ there exists a binary linear code family
that can correct up to a fraction $\delta$ of insertions and deletions and which has rate  at least $(1-\delta)/2-\ent(\delta)$.
\end{corollary}
Note that this corollary implies Theorem \ref{simplerandombinary}.

\begin{corollary}\label{cor:halfsingletonexistence}
For any $\delta > 0$ and $\epsilon > 0$ there exists a finite alphabet size $q=2^{O(\eps^{-1})}$ and a $q$-ary linear code family that can correct a fraction $\delta$ of insertions and deletions, and which has rate at least $(1-\delta)/2-\epsilon$.
\end{corollary}
This corollary implies the existential result part of Theorem \ref{halfsingletonboundupperandlower}.

A systematic code is one which encodes a message $x$ as the codeword $x \circ y$ where $y$ are some check symbols computed from $x$. That is, the message appears in raw form at the beginning of the codeword. 
For Hamming errors, any linear code can be transformed into a systematic code by doing a basis change of the generator matrix, followed by a possible permutation of the symbols. This transformation does not change the parameters of the code. For insertion and deletion errors, however, a permutation of the symbols may change the edit distance of the code, and hence it is not a priori clear that one can get a systematic linear code for insertion and deletion errors with the same parameters. Nevertheless, we now show that there exist such systematic linear codes with almost the same parameters.
The proof appears in Appendix~\ref{app:skipped-proofs}.

\begin{theorem}
\label{thm:existence-systematic}
Let $\F_q$ denote the finite field of size $q$. For any $\delta>0$ there exists a systematic $\F_q$-linear code family that can correct up to $\delta n$ insertions and deletions, with rate $(1-\delta)/2-\ent(\delta) / \log_2 q-o(1)$. Here $\ent(\delta)=-\delta \log_2 \delta-(1-\delta) \log_2(1-\delta)$ is the binary entropy function.
\end{theorem}

We can thus claim the analog of Corollaries~\ref{cor:binary-existence} and \ref{cor:halfsingletonbound} for systematic linear codes.

\iffalse
\begin{corollary}
For any $\delta>0$ there exists a binary systematic linear code $C: \F^m_2 \to \F^n_2$ that can correct up to $\delta n$ insertions and deletions, with rate $(1-\delta)/2-\ent(\delta)-o(1)$.
\end{corollary}

\begin{corollary}\label{cor:halfsingletonexistence}
For any $\delta > 0$ and $\epsilon > 0$ there exists a finite alphabet $q$ and a $q$-ary systematic linear code $C: \F^m_q \to \F^n_q$ that can correct up to $\delta n$ insertions and deletions, with rate $(1-\delta)/2-\epsilon$.
\end{corollary}
\fi

\section{Outer Bound for Rate of Linear Codes}
\label{sec:outer-bound}
To complement our existence results for linear insdel codes from Section~\ref{sec:bound} we prove here a novel upper bound on the rate of any linear code that can correct a given fraction of edit errors:

\begin{theorem}[Half-Plotkin Bound, restate of Theorem \ref{halfplotkin}]\label{thm:halfplotkinbound}
Fix a finite field $\F_q$. Every $\F_q$-linear insdel code which is capable of correcting a $\delta > 0$ fraction of deletions has rate at most $\frac{1}{2}(1 - \frac{q}{q-1} \delta) + o(1)$.
\end{theorem}

For binary codes this is a rate upper bound of $\frac{1}{2}(1 - 2\delta) + o(1)$. Furthermore since $\frac{q}{q-1} \geq 1$ for any $q$ the following Half-Singleton bound follows directly for any linear codes for edit errors, independent of alphabet size $q$:

\begin{corollary}[Half-Singleton Bound]\label{cor:halfsingletonbound}
Every linear insdel code which is capable of correcting a $\delta > 0$ fraction of deletions has rate at most $\frac{1}{2}(1 - \delta) + o(1)$.
\end{corollary}
This corollary implies the negative result part of Theorem \ref{halfsingletonboundupperandlower}.

The naming for these two bounds stems from the fact that they prove outer bounds on the rate of linear codes which are exactly half as large as the Plotkin bound of $(1 - \frac{q}{q-1} \delta) + o(1)$ and the alphabet-independent Singleton bound of $(1 - \delta) + o(1)$ which hold for codes (not necessarily linear) with relative Hamming distance $\delta$. Note that since two strings at fractional Hamming distance $\delta$ can lead to the same subsequence by deleting the $\delta$ fraction of symbols where they differ, the Plotkin and Singleton bounds also apply for general insdel codes.

Note that our Half-Singleton bound matches the existential bound for linear insdel codes given by Corollary~\ref{cor:halfsingletonexistence} proving that the best limiting rate achievable by a linear code over a finite alphabet that can correct up to a $\delta$ fraction of insdel errors is exactly $\frac{1}{2}(1 - \delta)$.

The proof for our Half-Plotkin bound is significantly more involved than in the Hamming case. The Half-Singleton bound of Corollary~\ref{cor:halfsingletonbound} has a much simpler direct proof which we first present as a warm-up. Both our proofs (implicitly) build on the $\frac{1}{2}$ rate upper bound of \cite{AGFC07} recorded below.

%given below as Lemma~\ref{lem:AGFChalfrate}.

\begin{lemma}[\cite{AGFC07}] \label{lem:AGFChalfrate}
Let $C$ be a linear code over any finite field $\F_q$ with message length $m$ and codeword length $n$. If $C$ can correct even a single deletion, then $m/n \leq 1/2$.
\end{lemma}

\begin{proof}[Proof of Corollary~\ref{cor:halfsingletonbound}, independent from Theorem~\ref{thm:halfplotkinbound}]
Let $C$ be a linear code with message length $m$ and codeword length $n$ which can correct any $\delta$ fraction of deletions. For any two different codewords $c_1,c_2 \in C$, we must have $\ED(c_1,c_2) \geq 2 \delta n+1$. Now consider the new code $C'$ obtained by removing the first $\delta n-1$ symbols of every codeword in $C$. For any two different codewords $c'_1, c'_2 \in C'$, we must have $\ED(c'_1, c'_2) \geq 3$. This means that $C'$ is a linear code capable to correct a single deletion. Thus Lemma~\ref{lem:AGFChalfrate} implies that $m /(n-\delta n+1) \leq 1/2$. Therefore $m/n \leq \frac{1}{2}(1-\delta)+\frac{1}{2n}$, giving our desired upper bound on rate. 
\end{proof}

We now return to the proof of the Half-Plotkin bound of Theorem~\ref{thm:halfplotkinbound}. The proof will make crucial use of the following technical lemma.

\begin{lemma}\label{lem:largedimensionimplieszeros}
Suppose $A \subseteq \F_q^{2d}$ is a subspace of dimension $t > 16 q$. Then $A$ contains a non-zero vector $v=(v_1,v_2)$ with both $v_1,v_2 \in \F_q^d$ containing at least $\frac{d}{q}(1 - 4 \sqrt{\frac{q}{t}})$ zeros.
\end{lemma}

We present the proof of the above lemma in Appendix~\ref{app:skipped-proofs} but elaborate here the main non-trivial aspect tackled by it. 
Note that a random vector in $A$ contains in expectation at least $\frac{2d}{q}$ zeros (since $A \subseteq \F_q^{2d}$ is a subspace of dimension $t$) and each of the halves $v_1$ and $v_2$ contains in expectation at least $\frac{d}{q}$ zeros. The crux of Lemma~\ref{lem:largedimensionimplieszeros} therefore lies in proving that there cannot exist correlations between the expectations for the two halves that are strong enough to prevent both expectation bounds to be satisfied simultaneously. Note that the $\frac{2d}{q}$ quantity of total number of zeros in a non-zero vector is tight (up to lower order terms) whenever $A$ is a good low-rate Hamming code. Indeed the minimum-Hamming distance of a linear Hamming code is exactly equal to the weight of the non-zero codeword with the lowest Hamming weight and codes matching (up to lower order terms) the $2d (1 - 1/q)$ Hamming-Distance given by the Plotkin bound are known. For such a good code, Lemma~\ref{lem:largedimensionimplieszeros} therefore guarantees 
a codeword with (up to lower order terms) minimal Hamming weight in which the $\frac{2d}{q} \pm o(1)$ zeroes are (up to lower order terms) exactly equally distributed between the two halves of the codeword.

\begin{proof}[Proof of Theorem~\ref{thm:halfplotkinbound}]
Suppose that $C$ is a $q$-ary linear insdel code of block-length $n$ which is capable of correcting any $\delta n$ deletions and suppose, for sake of contradiction that $C$ has a rate of at least  $(1 - \frac{q}{q-1} \delta)/2 + \epsilon$ for some $\epsilon >0$.

%We proof a stronger upper bound of $(1 - \frac{1}{\alpha(q)} \delta)/2 + o(1)$ for the rate of $C$ where $\alpha(q) \leq 1 - \frac{1}{q}$ is defined as follows:
%\[\alpha(q) = 1 - \lim_{\epsilon  \rightarrow 0}  \lim_{n \rightarrow \infty}   \min_{S \subseteq F_q^{2d}, dim(S) = \epsilon d}    \max \{LCS(x_1,x_2) \mid (x_1,x_2) \in (\langle S \rangle \setminus \{(0,\ldots,0)\}) \} / d.\]

%We first show that $(1 - \frac{1}{\alpha(q)} \delta)/2 + o(1)$ is indeed a valid upper bound and then show that $\alpha(q) \leq 1 - \frac{1}{q}$.

Let $C'$ be the linear code which encodes any input like $C$ but then deletes the first $d = \frac{q}{q-1} \delta n - \frac{\epsilon n}{5}$ symbols.

Let $C'_1$ be the linear code which encodes any input like $C'$ but then deletes the first symbol.

Let $C'_2$ be the linear code which encodes any input like $C'$ but then deletes the last symbol.

For sake of contradiction, assume that $\dim(C'_1 \cap C'_2) > \frac{\epsilon n}{5}$ and let $C'' \subseteq (C'_1 \cap C'_2) \subseteq \F_q^{n-d-1}$ be some linear sub-space of $\F_q^{n-d-1}$ with dimension $\frac{\epsilon n}{5}$ which does not contain the all-ones vector $(1,\ldots,1)$.

Note that by definition of $C_1'$, $C_2'$ and $C''$, every non-zero $c'' \in C''$ can be completed to two vectors $c = (v,\sigma,c'')$ and $c'= (v',c'',\sigma')$ in $C$, where $v,v' \in \F_q^{d}$ and $\sigma,\sigma' \in \F_q$. Note that $c$ and $c'$ cannot be identical because they each contain a copy of $c''$ albeit shifted by one symbol and $c''$ does not consist of only identical symbols since such vectors are excluded from $C''$. 

Consider a basis $c''_1,c''_2,\ldots \in C''$ of $C''$. For each such basis vector $c_i''$ fix one pair of completions $v_i,v_i' \in F_q^d$ and let $A \subset \F_q^{2d}$ be the subspace spanned by all such vectors, i.e., $A = \langle \{(v_i,v_i')  \mid i\} \rangle$. Note that each non-zero linear combination of basis vectors indeed gives rise to the same linear combination of its completion vectors, justifying the definition of $A$ as a subspace. Furthermore we claim that $\dim(A) = \dim(C'') = \frac{\epsilon n}{5}$, i.e., that the set of vectors $\{(v_i,v_i')  \mid 1 \le i \le \epsilon n/5\}$ is linearly independent. Indeed suppose that there is a non-trivial linear combination $c'' \neq 0$ of  basis vectors in $C''$ that completes to $(v,v') = 0$. In this case $C$ would contain two codewords starting with $d$ zeros followed by $n-d$ symbols containing a (shifted) copy of $c''$ plus one extra symbol (at the very end or after the $d$ zeros respectively). The edit distance between these two different codewords in $C$ is merely $2$, contradicting the assumption that $C$ can correct a much larger number of deletions. By Lemma~\ref{lem:largedimensionimplieszeros}, this means that there exists a $(v,v') \in A$ in which both $v$ and $v'$ contain a $1/q - o(1)$ fraction of $0$'s. The above implies the existence of codewords   $c,c' \in C$ starting with $v$ and $v'$ respectively, again contradicting the distance property of $C$ because $v$ and $v'$ can be transformed into an all-zero string using $(1 - 1/q + O_{q,\epsilon}(\frac{1}{n}))d = (1 - 1/q + o(1))d < \delta n - 2$ deletions while the latter parts of $c$ and $c'$ can be made equal using a single deletion. This completes the proof that $\dim(C'_1 \cap C'_2) \leq \frac{\epsilon n}{5}$.

We now have that 
\[ n - d - 1 \geq \dim(C'_1 \cup C'_2) = \dim(C'_1) + \dim(C'_2) - \dim(C'_1 \cap C'_2) \geq 2(\dim(C')-1) - \frac{\epsilon n}{5} \ , \]
and therefore
\[ \dim(C') \leq \frac{1}{2} (n-d+\frac{3\epsilon n}{5}) = \frac{1}{2} (1 - \frac{q}{q-1} \delta) n + \frac{2\epsilon n}{5} \ . \]   %\qedhere \]

Note that it is possible for the channel to take any sent codeword of $C$,
determine the most frequent symbol $\sigma$ in the the first $d$ coordinates, and delete all symbols in these $d$ coordinates but $\frac{d}{q}$ occurrences of $\sigma$, since $d(1-\frac{1}{q}) \leq \delta n$ is below the budget of allowed deletions. This means that after this operation all codewords in $C$ are still distinct. Furthermore, the span of the corrupted codewords received through this channel would form a subspace of dimension at most $\dim(C')+1$, since any corrupted codeword is of the form $\alpha^{d/q}$ followed by a codeword in $C'$, for some $\alpha \in \F_q$. This contradicts the assumption that $C$ has a rate of at least $\frac{1}{2}(1 - \frac{q}{q-1} \delta) + \epsilon$.
%Indeed, assuming dim(C'_1 \cap C'_2) = o(n) we have that
%n - d \geq dim(C'_1 \cup C'_2) = dim(C'_1) + dim(C'_2) - dim(C'_1 \cap C'_2) > 2(dim(C') -1) - o(n).
%and therefore
%dim(C') \leq (n-d)/2 + o(n) = (1/2) (1 - alpha delta) n + o(n).
\end{proof}

\section{Monte Carlo construction of linear insdel codes}
\label{sec:monte-carlo}

In this section we give a randomized Monte Carlo construction of linear codes s.t. with high probability over the randomness, we get an asymptotically good linear code for edit distance, with an efficient decoding  for a constant fraction of insdel errors.

The construction uses a random sequence where each element in the sequence is a string of $0$-symbols having length uniformly chosen from  $\{ 1, 2, \ldots, a \}$ where $a$ is an integer specified later. This sequence serves as a special kind of synchronization strings first introduced in \cite{haeupler2017synchronization}, then explicitly constructed in \cite{HS17c, CHLSW18}.

The high level idea is to use this sequence to cut a codeword $y$ for Hamming distance into blocks, and  argue that a carefully designed matching procedure, which is matching the edited codeword and a ``templet'', can   recover most of the symbols in $y$ as long as the total number of edit errors is bounded.

\subsection{Construction}

Let $C $ be an $(n_C, m, 2\kappa_C + 1)$ linear code over $\mathbb{F}_q $ that can correct $ \kappa_C $ Hamming errors.

Let  $S = (S[1], S[2], \ldots, S[n_C])$ be a sequence of independent random variables s.t. $S[i]$ is an all-$0$   string with length uniformly   chosen from $\{ 1,2,\ldots, a \}$, where $a = (\frac{12 n_C}{\kappa_C})^{20}$. 
Once chosen, both the encoder and the decoder will use the same $S$.

We construct an $(n, m, 2\kappa+1)$ linear code for edit distance over $\mathbb{F}_q$, with high probability it holds that $n = O(n_C)$ and $ \kappa=  0.01 \kappa_C$ insdel errors can be corrected.

\begin{construction}[Encoding]
\label{constr:renc}

\noindent The encoding operates as follows:
\begin{enumerate}
\itemsep=0ex
\item Input message $x \in \mathbb{F}_q^{m}$;
 
\item Compute $y = C(x) $;

\item Let
$$z = S[1] \circ y[1] \circ S[2] \circ y[2] \circ \cdots \circ S[n_C] \circ y[ n_C ];$$

\item Output codeword $z\in \mathbb{F}_q^n$.

\end{enumerate}

\end{construction}

\begin{construction}[Decoding]
\label{constr:rdec}
The decoding is as follows:
\begin{enumerate}
\itemsep=0ex

\item Input $z'\in \mathbb{F}_q^{n'}, n' = \Theta(n)$;

\item Let $z_? =  S[1] \circ ? \circ S[2] \circ ? \circ \cdots \circ S[n_C] \circ ? \in  \{0, ?\}^{n}$, where "$?$" is a special symbol different from elements in $\mathbb{F}_q$, indicating a blank to be filled, and $0$-symbol is the $0$-element in $\mathbb{F}_q$; 

\item Compute the $ ? $-to-non-zero matching $w$ between $z_?$ and $z'$, by Construction \ref{?to1matching};

\item Fill the blanks by using the matching $w$, to get $ y'$ which has a certain hamming distance from $y$; (Each blank is filled with the matched non-zero symbol, if there is no such a match, fill it with a $0$-symbol.)

\item Apply the decoding of $C$ on $y'$   to get $x$.

\item Output $x$.

\end{enumerate}

\end{construction}
Here we define the $ ? $-to-non-zero matching between two strings $z_? \in \{0,  ?\}^*$ and $z' \in \mathbb{F}_q^*$, to be a monotone matching, and further, every match in such a matching is a pair of indices $(i, j)$ that matches an $i$-th $?$-symbol of $z_?$, to a $j$-th non-zero symbol of $z'$.

Given a $?$-to-non-zero matching $w =( ( i_1 ,  j_1 ), \ldots, ( i_{|w|}  ,  j_{|w|} ))$, we define the cost function $ \mathsf{cost}(w) $ to be 
$$ \cost(w) =  \sum_{k=1}^{|w|} \idk\left(   p_{i_k} - p_{i_{k-1}} \neq q_{j_k} - q_{j_{k-1}}      \right) .$$
Here   $p_{i_k}$ is the position of the  $i_k$-th $?$-symbol in $z_?$, $q_{j_k}$ is the position of the $j_k$-th non-zero-symbol of $z'$.
$\idk $ is the indicator function.

Next we show our matching procedure which   returns a matching $w$ that maximizes  $\obj(w) = |w| - \mathsf{cost}(w)$.
\begin{construction}[$?$-to-non-zero Matching Procedure]
\label{?to1matching}
On input strings $ z_? $ and $ z' $, the procedure returns a $?$-to-non-zero matching $w$.

\begin{itemize}

\item
Let $n_1$ denote the number of non-zero-symbols in $z'$.

\item
Let $p_i$ denote the position of the $i$-th question mark in $z_?$, and $q_j$ denote the position of the $j$-th non-zero in $z'$.

%Let $\prevone: [n'] \rightarrow [n']$ denote the function that returns the index of the immediate previous (according to sequential order $1,2, \ldots, n'$ of string $ z'$ ) symbol non-zero.

\item
We want to use $f[i][j], i\in [n_C], j\in [n_1]$ to record the maximum $ \obj(w) $, among all $w$, each being a $ ? $-to-non-zero matching  between $z_?[1, p_i]$ and $z'[1, q_j]$ with the last match being $ (i, j) $.
%$f[i][0] = 0, f[0][j] = 0$ for every $i, j$.
\end{itemize}

\begin{enumerate}

%\item Initialize $ f[0][j] = 0, j\in [n_1]$,    $ f[i][0]  = 0, i\in [n_C]$;

\item Initialization: $ f[i][j] = 0,  i\in \{ 0,1,\ldots,  n_C \}, j\in \{0,1,\ldots, n_1\} $;

\item Transition step: For $i = 1$ to $ n_C $, $j = 1$ to $n_1$, 
$$ f[i][j]  =  \max_{i'< i, j'< j}\left\lbrace    f[i'][j']  + \idk\left(  p_i - p_{i'} = q_{j} - q_{j'}  \right)   \right\rbrace, $$
where if there is no such $ (i', j') $, then 
$$ f[i][j]  =   \idk\left(  p_i   = q_{j}  \right); $$

\item Output  $ \max_{i,j}\lbrace    f[i][j]   \rbrace.$ 

\end{enumerate}

\noindent
To output the corresponding matching, we only need to store the corresponding matching along with each time we compute an $f[i][j] $. For the initialization step, those matchings are all empty. For the transition step,
\begin{itemize}

\item  if $   p_i - p_{i'} = q_j - q_{j'}  $, then the corresponding matching to $f[i][j]$ is the concatenation of the matching for $f[i'][j']$ and the match $ (i, j) $, where $i', j'$ are the indices we pick when taking the maximum; 

\item else the corresponding matching to $f[i][j]$ is an empty matching.

\end{itemize}

\end{construction}

\subsection{Analysis}

\begin{lemma}
\label{lem:DPcorrect}
The $?$-to-non-zero matching procedure returns a matching which has a maximum $\obj $.

\end{lemma}

\begin{proof}

We use induction to show that $f[i][j]$ records the maximum $\obj(w)$ among all matchings  between $z_?[1, p_{i}]$ and $ z'[1, q_{j}] $ whose  last match is $(i, j)$.

The base case is $f[1][j], j\in [n_1], f[i][1], i\in [n_C]$. $f[1][j]  $ is one if   $p_1 = q_{j}  $, otherwise it's $0$.  $f[i][1] $ is one if $p_i = q_1$, otherwise it's $0$. Notice that the transition function indeed achieves this.

For the induction case, in order to compute $f[i][j], i\in [n_1], j\in [n_C]$, we assume every 
$f[i'][j'], i'<i, j'<j$ and its corresponding  matching with maximum target function value has already been computed. 

Let $w$ be the matching  attained when $f[i][j]$ is computed. So $f[i][j] = |w| -\cost(w)$.
Notice that $w$ can be divided into two parts. 
The first part $w_1$ includes every match of $w$ except for the last one. So it is a matching between $z_?[0, p_{i'}]$ and $ z'[0, q_{j'}] $, where $( {i'},  {j'})$ is the second last match in $w$ and also the last match in $w_1$. 
The second part is the last match $ ( i,  j) $ of $w$.

We claim that $ \obj(w_1) $ has to be $f[i'][j']$.  Suppose  it is not. Then it cannot be larger because  $f[i'][j']$ is the maximum by the induction hypothesis. So it can only be smaller. But then we can use the matching $w'_1$ corresponding to $f[i'][j']$ to replace the first part and get another matching $w' = w'_1 \circ ( i,  j)$   s.t. 
\begin{equation}
\label{eq:targetvalueinduction}
\begin{split}
|w'|-\cost(w') & = |w'_1| - \cost(w'_1) + \idk\left(   p_i - p_{i'} = q_j - q_{j'}  \right)\\
& > |w_1| - \cost(w_1) +\idk\left(    p_i - p_{i'} = q_j - q_{j'} \right)\\
& = |w|-\cost(w).
\end{split}
\end{equation}
The first equation is because of the definition of $\cost$ and the structure of $w'$.
The inequality is because our picking of $w'_1$.
The last equation is because the definition of $\cost$.
However (\ref{eq:targetvalueinduction}) contradicts the definition of $f[i][j]$.
This shows the claim.

As a result,  $ f[i][j] $ has to be in the form 
\[  \max_{i'< i, j'< j}\left\lbrace   f[i'][j']  + \idk\left(  p_i - p_{i'} = q_j - q_{j'} \right)   \right\rbrace.
\] 
This is exactly what the procedure computes. So it can compute $f[i][j]$ and get the corresponding matching correctly.
This shows the induction step.
\end{proof}

\begin{lemma}
\label{lem:notmatchedones}
The number of non-zeros in $z'$ that are not matched, is at most $3\kappa$.

\end{lemma}

\begin{proof}

Consider the intuitive matching $w^*$ which is induced by matching those non-inserted non-zero-symbols in $z'$  to their original positions in $z$. 
The original position of a symbol means the index of this symbol in $z$. 
Since $z_?$ and $z$ has the same length, these original positions are also one-on-one correspond to positions in $z_?$.
When there are no edits,  $\obj(w^*)$  is equal to $ \tau$, which is the number of non-zero-symbols in $z$.
This is because as there are no edits, for every $(i, j) \in w^*$, $p_{i} - p_{i-1} = p_{j}- p_{j-1}$. 
Thus $\cost(w^*) = 0$ by its definition and hence  $\obj(w^*) = \tau - \cost(w^*) =\tau $, where $\tau$ is the weight of $ z$.

Now we analyze the effects of $\kappa$ insdel errors on $\obj$.
Each insertion can
\begin{itemize}
%\item change the number of consecutive $0$s right before a non-zero-symbol;
\item insert a $0$-symbol and thus increase the number of consecutive $0$-symbols immediately before a non-zero-symbol by one;

\item  insert a   non-zero-symbol.
\end{itemize}
Both bullets can decrease $\obj$ by at most one. Because they can increase $\cost$ by at most one, as each insertion may change the distance between the two matched symbols immediately before and after that inserted symbol.

Each deletion can
\begin{itemize}
 
\item delete a $0$-symbol and thus decrease the number of consecutive $0$-symbols immediately before a non-zero-symbol by one;  

\item delete a non-zero-symbol. 
\end{itemize}  
The first bullet will decrease $\obj$ by one as it changes the distance between the two matched symbols  immediately before and after that deleted symbol, and thus increasing $\cost$ by at most one.
The second bullet will decrease $\obj$ by $\leq 2$ since it can cancel a possible match and also changes the distance between the two matched symbols immediately before and after that deleted symbol.

Thus $\kappa $ insdel errors can decrease the target function value by at most $2\kappa $.
As the algorithm returns $w$ which has a maximum $\obj(w)$, we know $\obj(w) \ge \tau - 2\kappa$.
Notice that the weight of $w$ is $ \le   \tau + \kappa$.
Hence the number of unmatched non-zeros in $z'$ cannot be more than $3\kappa$.
\end{proof}

Next we  show that the decoding can correct $\kappa  $ insdel errors.

\begin{definition}
A match $( i,  j)$ is bad, if the $j$-th non-zero symbol of $z'$ is the  $i_j$-th symbol of $z$,  but $i_j \neq i$. 

\end{definition}
Note that there are some non-zero-symbols in $z'$ which are inserted, we will not consider them here i.e. the badness are only defined over matches whose second entry points to a symbol of $z'$, which is already in $z$ before insertions and deletions.
We denote $\Delta(w)$ as the number of bad matches in $w$. When $w$ is clear in the context, we simply use $\Delta$.

\begin{lemma}
\label{lem:badmatchcount}
Suppose $\HD(y,y') =t > \kappa_C$, then  $\Delta(w)  > (\kappa_C - 6\kappa)/2  $.

\end{lemma}

\begin{proof}

The $t$ Hamming errors come  from two different types of matches, which are
\begin{enumerate}

\item an $i$-th question mark is not matched, but $y[i] \neq 0$;

\item an $i$-th question mark is matched, but $y[i] = 0$. 

\end{enumerate}
We denote the number of first type   errors to be $t_1$, and   the number of  second type errors to be $t_2$. 
So 
\begin{equation}\label{ineq:t}
t_1 + t_2 = t > \kappa_C.
\end{equation}

For $t_2$, the errors are only from two sources, i.e. the non-zero-symbol in such a match is   from  (1)  inserted non-zero-symbols or (2) non-zero-symbols in $y$. 
Each inserted non-zero-symbol may be in one of such matches, contributing to one such error. 
Each non-zero symbol in $y$  may be matched to   a question mark which should  a $0$, contributing to one such error,
noticing that this can only happen when the match is bad.
Hence,
\begin{equation}\label{ineq:t2}
t_2 \leq \kappa + \Delta.
\end{equation}

For $t_1$, 
we are actually counting the number of non-zero symbols in $z$ that are not matched, because $y$ is a subsequence of $z$ which contains all the non-zeros in $z$. 

We claim that $t_1$ is upper bounded by the summation of 
\begin{enumerate}

\item the number of non-zeros that are deleted from $z$,

\item the number of non-zeros in $z'$ that are not matched,

\item the number of non-zeros in $z'$ that are matched to some $i$-th question mark but $y[i] = 0$.

\end{enumerate}

This is because $t_1$ is the number of not matched non-zero symbols in $z$. So it pluses the number of matched non-zero-symbols in $z$, is equal to the total number of non-zero-symbols in $z$. 
\begin{align}
\label{involvedeq1}
& t_1 + \mbox{ the number of matched non-zero-symbols in } z\\
= & \mbox{ the total number of non-zero-symbols in } z.     
\end{align}
On the other hand, the total number of non-zero-symbols in $z$ is at most the number of non-zero-symbols in $z'$ pluses the number of deleted non-zero-symbols (term 1). And the number of non-zero-symbols in $z'$ is the summation of  term 2, term 3 and   the number of matched non-zero-symbols in $z$.
This is because    the number of non-zeros in $z'$ that are matched to some $i$-th question mark and $y[i] \neq 0$ is equal to  the number of matched non-zero-symbols in $z$.
Hence
\begin{align}
\label{involvedeq2}
&\mbox{ the total number of non-zero-symbols in } z\\
= &\mbox{ term 1 } +  \mbox{ term 2 } + \mbox{ term 3 } 
+  \mbox{ the number of matched non-zero-symbols in } z.
\end{align}
By equality \ref{involvedeq1} and \ref{involvedeq2}, $t_1$ is bounded as desired.

The number of non-zero-symbols deleted can be at most $\kappa$, since there are at most $\kappa$ deletions. 
By Lemma \ref{lem:notmatchedones} the 2nd term is at most $3\kappa$.
The third item is at most $t_2$ by their definitions.
So
\begin{equation}\label{ineq:t1}
t_1 \leq 4\kappa + t_2.
\end{equation}

By inequality (\ref{ineq:t}), (\ref{ineq:t2}), (\ref{ineq:t1}), 
\[
\Delta > (\kappa_C - 6\kappa)/2.
\]
\end{proof}

For a bad match $(i, j)$, recall that $i_j$ be the original position in $z$ of the $j$-th non-zero-symbol of $z'$.
Note that since we only define bad match on some $j$ which points to a non-zero-symbol that is not deleted from $z$, thus there is always an $i_j$ for $j$.

An interesting property of those matches in $w$ that are not involved with inserted symbols, is that
they  induce a self-matching between $z_?$ and $z_?$.
We denote this matching as $ \tilde{w} $, i.e. each match $(i, i_j)$ of $\tilde{w}$ is a match $(i, j)$ in $w$ s.t. $j$ points to a non-zero symbol which is not deleted from $z$.
Since bad matches are only defined for matches whose second entries are non-zero symbols of $z$, $\Delta(\tilde{w} ) = \Delta(w)$.

\begin{lemma}
\label{lem:costlarge}
With probability at least $ 1- \left( \frac{e n_C}{\Delta}\right)^{2\Delta} a^{- 0.1 \Delta}$ over $S$, for any self-matching $\tilde{w}$ on $z_?$ having $\Delta$ bad matches,
there are at most $ 0.1\Delta $ bad matches   $(i, j)$ s.t. $p_{i}-p_{i'} = p_{j} - p_{j'} $, and hence
$\cost(\tilde{w}) \geq 0.9\Delta.$
\end{lemma}
Notice that the probability is meaningful when both $a$ and $\Delta$ are large enough and hence $a^{-\Theta(\Delta)}$ is the dominating term. We will see later that this is indeed the case in our setting of parameters.

\begin{proof}
Consider a specific sequence of $\Delta $ bad matches. 
We first claim that with probability $\geq 1-a^{-0.1\Delta }$, for any matching $\tilde{w}$ containing this specific sequence of bad matches, there are
$\Delta' <  0.1\Delta $ bad matches $(i, j)$ s.t. $p_{i}-p_{i'} = p_{j} - p_{j'} $.

For a bad match $(i, j)$ in the sequence, if it is the first match in the matching, then $p_{i} \neq p_j$ since $|S[i]|, |S[j]| > 0$.
If it has previous match $(i', j')$, then $(i', j')$
is either a bad match or a good (not bad) match. 
If it is a good match, then $p_{i}-p_{i'} \neq p_{j}-p_{j'}$, since  $|S[k]| > 0, \forall k$.
If it is a bad match, then 
  conditioned on $S[1], \ldots, S\left[ \max\left( i, j \right) - 1 \right]$ being arbitrary fixed values, 
$ p_{i}-p_{i'} = p_{j} - p_{j'} $ happens with probability at most $1/a$.
Because,  in this condition, $p_{i}-p_{i'} = p_{j} - p_{j'} $ happens only if $S\left[ \max\left( i, j \right)\right]$ is equal to one specific value in $[a]$.

So for any matching $\tilde{w}$ including this specific sequence of bad matches,  the probability that there are  $  \geq  0.1\Delta $ bad matches $(i, j)$ s.t. $p_{i}-p_{i'} = p_{j} - p_{j'} $, is at most $a^{- 0.1\Delta }$.

The number of sequences of  $\Delta$ bad matches is at most 
$$  { n_C \choose \Delta }{n_C \choose \Delta} \leq \left( \frac{e n_C}{\Delta}\right)^{2\Delta}.  $$

So by a union bound,   with probability $ 1- \left( \frac{e n_C}{\Delta}\right)^{2\Delta} a^{- 0.1 \Delta}$, there are at most $ 0.1\Delta $ bad matches s.t.  $(i, j)$ s.t. $p_{i}-p_{i'} = p_{j} - p_{j'} $, and hence $\cost(\tilde{w}) \geq (1- 0.1)\Delta =  0.9\Delta$.
\end{proof}

%\begin{lemma}
 
%For any matching $\tilde{w}$ between $z_?$ and itself, if there are $\Delta$ bad matches in $\tilde{w}$,  
%then with probability $ 1- a^{-\Theta(\Delta)}$, $\cost(\tilde{w}) \geq 0.9\Delta$.

%\end{lemma}

%\begin{proof}
%For a bad match $(i, j)$, we claim that conditioned on $S[1], \ldots, S\left[ \max\left( i, j \right) - 1 \right]$ being arbitrary fixed values, 
%$ p_{i}-p_{i'} = p_{j} - p_{j'} $ happens with probability at most $1/a$.
%This is because,  in this condition, $p_{i}-p_{i'} = p_{j} - p_{j'} $ happens only if $S\left[ \max\left( i, j \right)\right]$ is equal to a specific value. Thus this happens with probability at most $1/a$. 

%So the probability that there are  $\Delta' \geq  0.1\Delta $ bad matches $(i, j)$ s.t. $p_{i}-p_{i'} = p_{j} - p_{j'} $, is at most $a^{-\Theta(\Delta')}$.
%By a union bound over all $\Delta'$, we know with probability $ 1- a^{-\Theta(\Delta)}$, there are at most $ 0.1\Delta $ bad matches s.t.  $(i, j)$ s.t. $p_{i}-p_{i'} = p_{j} - p_{j'} $, and hence $\cost(\tilde{w}) \geq 1- 0.1\Delta =  0.9\Delta$.

%\end{proof}

\begin{lemma}
\label{lem:costwlarge}
For $S$ being any string in $\supp(S)$, for any $\kappa$ insdels, $\cost(w) \geq \cost(\tilde{w}) - \kappa $.

\end{lemma}

\begin{proof}
Recall that $\tilde{w}$ consists of all $(i, i_j)$ s.t. $(i, j)$ is a match in $w$ and $j$ points to a non-inserted symbol.
Each insdel will only decrease cost function by $1$, since for each match $(i, j)$, $\mathbbm{1}(p_{i} - p_{i'} = q_{j} - q_{j'})$ only involves $(i, j)$ and its immediate previous match.
As there are at most $\kappa$ insdels, the lemma holds.
\end{proof}

\begin{lemma}
\label{lem:rcorrect}
With probability $ 1  -2^{-0.9\kappa_C}$ over $S$,   the encoding and decoding   give a code which can correct $\kappa    $ insdels.

\end{lemma}

\begin{proof}
Suppose the decoding cannot compute $x$ correctly. Then it must be the case that $ \HD(y, y') > \kappa_C  $. Hence by Lemma \ref{lem:badmatchcount}, $\Delta(w)> (\kappa_C - 6\kappa)/2$. 
Then by definition of $\tilde{w}$, $\Delta(\tilde{w}) = \Delta(w) >(\kappa_C - 6\kappa)/2 $.
By Lemma \ref{lem:costlarge}, with probability 
$
   1- \left( \frac{e n_C}{\Delta}\right)^{2\Delta} a^{- 0.1 \Delta}
$
over $S$, for any $\tilde{w}'$ of $\Delta$ bad matches, we should have $\cost(\tilde{w}') \geq 0.9\Delta$.
Notice that since $ \Delta > (\kappa_C - 6\kappa)/2, a = \left( \frac{12 n_C}{k_C} \right)^{20}$, the probability is at least $1-2^{-0.9\kappa_C}$.
For this $\tilde{w}$, we should also have  $\cost(\tilde{w}) \geq 0.9\Delta$.
Hence by Lemma \ref{lem:costwlarge}, $\cost(w) \geq 0.9\Delta - \kappa \geq 0.9(\kappa_C - 6\kappa)/2 - \kappa > 2\kappa$, recalling that $\kappa = 0.01 \kappa_C$. 
So $\obj(w) = |w| - \cost(w) \leq |w| - 2\kappa$.
However, by Lemma \ref{lem:DPcorrect}, $\obj(w)  $ is maximized, which is at least $\tau - 2\kappa$. By definitions of $\tau$ and $w$, this is at least $ |w| - 2\kappa$. 
A contradiction is reached.
\end{proof}

Next we show that the codeword length is $O(n_C)$.
\begin{lemma}
\label{lem:rcodewordlen}
For any $\gamma>0$,
with probability at most $2e^{-  \gamma^2   n_C(a+1) /6}$, $$ n  \in \left[ (1-\gamma)  n_C   (a+1)/2 +n_C, (1+\gamma)  n_C   (a+1)/2   +n_C \right].$$

\end{lemma}

\begin{proof}

Since the length of $S[i]$ is uniformly chosen to be $ 1,\ldots, a $, 
the expectation of their total length is $ n_C(a+1)/2  $.
By a Chernoff bound, with probability at least  $1- 2e^{-  \gamma^2   n_C(a+1) /6}$, $\sum_{i \in [n_C]} |S_i| \in [  (1-\gamma) n_C(a+1)/2, (1+\gamma) n_C(a+1)/2   ]$. 
So 
\[
n = n_C+  \sum_{i \in [n_C]} |S_i| \in  \left[   (1-\gamma)  n_C(a+1)/2  +n_C, (1+\gamma)   n_C(a+1)/2  +n_C \right ].  \]
\end{proof}

The efficiency of the encoding and decoding follows directly from the construction. 
\begin{lemma}
\label{lem:rtime}
The encoding is in linear time of the encoding time of $C$;
The decoding is in time $O(n^4)$ plus the decoding time of $C$.
\end{lemma}
We remark that the $O(n^4)$ time in decoding is because of the dynamic programming matching procedure.

\begin{theorem}
\label{thm:linearcoderandconstruct}
If there is an explicit $q$-ary linear $(n_C, m , 2\kappa_C +1)$ code for Hamming distance, then for any $ \gamma >0$, there is a $q$-ary linear $ (n, m, 2\kappa+1) $ code for edit distance, with $  n  \in \left[  (1-\gamma)  n_C(a+1)/2  +n_C, (1+\gamma)  n_C(a+1)/2    +n_C \right], \kappa = 0.01\kappa_C $,    $a = (\frac{12 n_C}{\kappa_C})^{20} $. 
The encoding is randomized with failure probability at most $ 2^{- 0.9 \kappa_C }  + 2e^{- \gamma^2   n_C(a+1)/6 } $ over the random choice of $S$. 
The decoding is explicit when given $S$, correcting $\kappa$ insdels.

\end{theorem}

\begin{proof}

Construction \ref{constr:renc} \ref{constr:rdec} gives such a code. 
Message length and the unique decoding radius directly follows from the construction. 
Codeword length follows from \ref{lem:rcodewordlen}.
Correctness follows  Lemma \ref{lem:rcorrect}.
Computing efficiency follow from \ref{lem:rtime}.
The success probability follows    \cref{lem:rcorrect} and \cref{lem:rcodewordlen} by a union bound.

Note that the code constructed is linear.  Since by the construction, $C$ is a linear code and the encoding is just inserting some number of zero-symbols to some fixed positions of $y = C(x)$. So taking arbitrary two   codewords $ z_1, z_2 $ corresponding to $y_1, y_2 \in C$, their summation is still a codeword $z$ which corresponds to $y_1 + y_2 \in C$. Also, taking any codeword $z$, multiplying every coordinate with $ \sigma $ in $\mathbb{F}_q$, the result is still a codeword, because the multiplication does not change the zero-symbols and for   $y$, multiplying $\sigma$ with every coordinate results a  codeword in $C$ since $C$ is a linear code.
So the constructed code is linear.
\end{proof}

We state the following theorem  to explicitly describe the error rate and information rate. It immediately follows from Theorem \ref{thm:linearcoderandconstruct}, by taking $\gamma$ to be a constant.
\begin{theorem}
If there is an explicit linear  code for Hamming distance with block length $n_C$, error rate $\delta_C$, information rate $ \gamma_C $, then there is a linear code for  insdel errors, with block length $   O( n_C/\delta^c_C ) $, error rate $\delta = O(\delta^{1+c}_C)$, information rate $ O(\delta^c_C \cdot \gamma_C)  $, for some large enough universal constant $c$.
The encoding is randomized with failure probability at most $ 2^{-\Omega( \delta_C n_C )} $ over the random choice of $S$. 
The decoding is explicit when given $S$, correcting $\delta$ fraction of insdel errors.
\end{theorem}
Notice that this directly implies Theorem \ref{thm:linear-const-intro}, since there are various explicit constructions of asymptotically good linear code for Hamming distance.

\section{Derandomization and explicit construction of linear insdel codes}
\label{sec:derand}
In this section, we derandomize the encoding of the previous section.
The key idea is to develop a special string  $s$ to replace the random sequence $S$.

\subsection{Synchronization Separator Sequence}

We   use the following terminologies about self-matchings for strings consists of $0$s and question marks.

For a sequence $ S[1] \circ ? \circ S[2] \circ ? \circ \cdots \circ S[n ] \circ ? $, a  $?$-to-$?$ self matching $w$ is a monotone matching between two substrings of this sequence s.t. each match $(i, j)$ matches the $i$-th question mark to the $j$-th question mark. 
\begin{definition}
A match  $(i, j)$ in such a matching is called undesired if 
\begin{itemize}
\item $i\neq j$ (i.e. bad);
\item $  p_i-p_{i'} = p_j - p_{j'}$, when $(i, j)$ is not the first match.
Here $(i', j')$ is the immediate previous match of $(i, j)$ in the matching. 
\end{itemize}

\end{definition}

Next we show the string defined as follows can be used to replace the random sequence $S[1], S[2], \ldots, S[n]$.

\begin{definition}[$(\Lambda, a)$ synchronization separator sequence]

$s[1], s[2], \ldots, s[n]$ is called a $(\Lambda, a)$ synchronization separator sequence, if for any self-matching of $ s[1] \circ ? \circ s[2] \circ ? \circ \cdots \circ s[n ] \circ ?   $, the number of undesired matches is at most $\Lambda$.

Here each $s[i], i\in [n]$ is an all $0$ string with length in $\{ 1,\ldots,a \}$.
\end{definition}

We show that synchronization separator sequences can be constructed explicitly.

\begin{lemma}
\label{lem:syncSepSequenExplicit}

There is an explicit construction of $ (\Lambda, a) $ synchronization separator sequences of length $n$, for every    $n\in \mathbb{N}$, $\Lambda \in \mathbb{N}$, constant $c \ge 3$,    $a =  (\frac{en}{\Lambda})^{c} $ being a power of $2$.

\end{lemma}

To show the lemma, we give the following construction.
\begin{construction}
\label{constr:syncSepSequenExplicit}

\begin{enumerate}

\item 
Let $ g:\{0,1\}^{d_g} \rightarrow \{0,1\}^{n_g = n\log a}$ be an $ \varepsilon_g$ almost $ n_g $-wise independence generator, with $d_g = O( \log \frac{n_g \log n_g} {\varepsilon_g^2} )$, $\varepsilon_g = a^{-\Delta'}$, from Theorem \ref{almostkwiseg}, where $\Delta' =  12  \frac{\log n }{ \log \frac{n}{\Lambda} }$;

\item For every $ r\in \{0,1\}^{d_g} $, 

\begin{enumerate}
\item compute $   g(r)$, and then partition it into a sequence of length $ \log a $ blocks, where the $i$-th block, $i\in [n]$, corresponds to a binary number $a_i \in \{ 1,\ldots,a \}$;

\item
Generate sequence $s = s[1] \circ ? \circ s[2] \circ ? \circ \cdots \circ s[n ] \circ ?$ s.t $ \forall i\in [n], s[i] = \underbrace{0\circ 0 \circ \cdots \circ 0}_{a_i}$;

\item For some large enough constant $c$, for every pair of substrings $u, v$ of $s$,  s.t.  $|u| + |v| \leq   \frac{ 12 n}{\Lambda} \frac{\log n }{ \log \frac{n}{\Lambda} } $,  compute a matching   which only consists of bad matches, having size $ \Delta'  $ and the maximum number of undesired matches; 

\item Let   $\Lambda_0$ be the maximum number of undesired matches, among all the matchings computed in the above step;

\item If  $\Lambda_0 < 6 \frac{\log n }{ \log \frac{n}{\Lambda} }  $, then halt and output $s$.

\end{enumerate}

\item Abort.

\end{enumerate}

\end{construction}
Now we show a property of undesired matches and then use it to prove \cref{lem:syncSepSequenExplicit}.
\begin{lemma} 
\label{lem:selfMatchingCut}
If there is a self-matching $w $ between two substrings $u, v$ of $s$, having $\Lambda$ undesired matches, then for every $i \in \mathbb{N}$, there is a self-matching $w'$, having at most $ \Lambda' = \Lambda/2^{i} $ undesired matches, between two substrings $u', v'$, where $u'$ is a substring of $u$, $v'$ is a substring of $v$, and $\frac{\Lambda'}{ |u'| + |v'| } \geq \frac{\Lambda}{|u|+|v|}$. 
\end{lemma}
We remark that the proof strategy is similar to that of \cite{KZXK18}. 
The difference is that in \cite{KZXK18} the object of interest was the matching size. But here we care about the number of undesired matches in the matching. We cannot just pick out these undesired matches to form a new matching to consider, because the ``undesired'' property depends on other matches  which may not be an undesired matches.
\begin{proof}[Proof of \cref{lem:selfMatchingCut}]
We repeatedly use the following cutting technique on the current matching $w$. We denote ratio $ \gamma = \frac{\Lambda}{l} $ and length $l = |u|+|v|$.

We find the match $(i, j)$ s.t. the number of undesired matches from the first match to $(i, j)$ (include), is $ \Lambda_1 = \lfloor \Lambda/2 \rfloor$.  Thus the number of undesired matches from $(i,j)$ (not include) to the end of $w$ is   $ \Lambda_2 = \Lambda - \lfloor \Lambda/2 \rfloor$. Here $(i,j)$ also divides $w$ into two halves. The first half matching $w_1$ is from the first match to $(i, j)$ (include). The second half from $(i, j)$ (not include) to the last match.
Note that $u, v$ are also divided by $(i, j)$. Let $l_1$ be the length of the summation of the lengths of first halves of $u, v$, including the symbols that $i ,j $ pointed to. For the second halves it's $l_2$, not including the symbols that $i, j$ pointed to. We compare $ \Lambda_1/l_1 $ and $\Lambda_2/l_2$. Pick the larger one, and then consider the corresponding halves of $w$ and $u, v$.

The procedure ends when we get a pair of  strings $u, v$ and a  matching $w$ between them s.t. the number of undesired matches is at most $\Lambda'$. 

Notice that $ \Lambda_1 + \Lambda_2 = \Lambda $, $ l_1 + l_2 = |u| + |v| $. Let $\gamma_1 = \Lambda_1/l_1$, $\gamma_2 = \Lambda_2/l_2$. Then one of $\gamma_1$ and $\gamma_2$ has to be at least $\gamma$. Otherwise, $\Lambda/l = \frac{\Lambda_1+\Lambda_2}{ l_1 + l_2 }$ has to be smaller than $\gamma$, causing a contradiction. So the one we pick has ratio at least $\gamma$. 
Also the number of undesired matches is divided into two halves each time we apply a cut. 

As a result, after doing cutting for $i$ times, we get the desired matching and substrings.
\end{proof}
Next we prove Lemma \ref{lem:syncSepSequenExplicit}.
\begin{proof}[Proof of \ref{lem:syncSepSequenExplicit}]
We use Construction \ref{constr:syncSepSequenExplicit}.

We show that there is a seed $r\in \{0,1\}^{d_g}$ which can let Construction \ref{constr:syncSepSequenExplicit} output the $s = s(r)$, which is a synchronization separator sequence we want.

The proof strategy is to consider a uniform random seed $r$ and show that with some probability  the construction can output a desired $s$.

If there is a   self-matching $w$ having $\Lambda$ undesired matches,
then by Lemma \ref{lem:selfMatchingCut}, there is a  self-matching $w' $ between some $u'$ and $v'$ which are substrings of $s$, s.t. there are $\Lambda' =  \Lambda/2^i = 6 \frac{\log n}{\log \frac{n}{\Lambda}} $ undesired matches in $w'$, for some large enough $i$, where $|u'|+|v'| = l' \leq \frac{12 n}{\Lambda} \frac{\log n }{ \log \frac{n}{\Lambda} } $. 
Also we can think of $w'$ as only consisting of bad matches since good matches cannot be undesired by definition and for any undesired match, its immediate previous match has to be a bad match due to $s[k] > 0, \forall k$.
Also we only need to consider $w'$ of size $\Delta' = 2\Lambda'$ since there is always a size $w'$ matching which can contain the $\Lambda'$ undesired bad matches.
 
On the other hand, for arbitrary $u , v$ in the algorithm, we first consider that $s[1],\ldots, s[n]$ are at uniform random.
The probability that a undesired match $(i, j)$ in $w'$ happens, conditioned on fixing $s[1], \ldots, s\left[\max\left(i , j\right) - 1\right]$ is at most $1/a$,  since we need $s\left[\max\left(i, j\right)\right]$ to take a specific value to make $p_i-p_{i'} = p_j - p_{j'}$.
So the probability that a specific matching   with $ \Lambda' $ undesired matches, happens with probability $\le a^{-\Lambda'}  $. 

Next we consider using $\varepsilon_g$ almost $n$-wise independence.  Each sample point which happens with probability $\rho \ge a^{-\Delta'}$ in the uniform random case instead happens with probability at most $\rho + \varepsilon_g \leq 2\rho$ when $\varepsilon_g  = a^{-\Delta'}$.
Hence the probability that a specific matching with $ \Lambda' $ undesired matches, happens with probability at most  $2a^{-\Lambda'}$.

There are at most $n^4$ such pair of $u, v$. 
For each pair
the number of matchings $w'$ between them is  at most $ {l' \choose \Delta'} = \left( \frac{e l'}{\Delta'} \right)^{\Delta'} \le \left( \frac{e n }{ \Lambda } \right)^{\Delta'}$ since we only consider $w'$ with size $ \Delta'$.
So by a union bound, with probability $ 1-  n^4 \left( \frac{e n }{ \Lambda } \right)^{\Delta'} 2 a^{-\Lambda'}$, there is no size $\Delta'$  matching between any pair of substrings $u,v $ with $l'  $  having $\Lambda'$ undesired matches. 
When this happens, there is no self-matching with $\Lambda$ undesired matches for the   string $s $.
Since $a = \left(\frac{e n}{\Lambda}\right)^c, c\ge 3$, 
the probability is $\ge 1- 2n^4 \left( \frac{en}{\Lambda} \right)^{-\Lambda'} \ge 1 - 1/n$.

We can find a seed $r$ which gives us such a string $s(r)$, in polynomial time. This is because  $|r| = O(\log n)$ and in the above computation  the number of testings is $\poly(n)$. The testings are also efficient.  Thus an exhaustive search for a $(\Lambda, a)$ synchronization separator sequence can be done in polynomial time.   
\end{proof}

\subsection{Deterministic Encoding}

\begin{theorem}

\label{thm:linearcodedetermconstruct}

If there is an explicit      linear $q$-ary $(n_C, m, 2\kappa_C + 1)$ code for Hamming distance, then there is an explicit       linear  $q$-ary$(n, m, 2\kappa+1)$ code for edit distance, where $ n = (a+1)n_C$,  for   $ a =(  \frac{5e n_C}{\kappa_C})^{ 3}$, $ \kappa = 0.01 \kappa_C$.
\end{theorem}

\begin{proof}

We replace the random sequence $S[1], S[2], \ldots, S[n_C]$ in Construction \ref{constr:renc} by a $(\Lambda', a)$ synchronization separator sequence from Lemma \ref{lem:syncSepSequenExplicit}.
Here we let $\Lambda' =0.2 \kappa_C  $, $ a =  \left( \frac{e n_C}{\Lambda'} \right)^{3}$.

Suppose  $\HD(y, y') > \kappa_C$.
By Lemma \ref{lem:badmatchcount} 
   the number of   bad matches in $w$ is at least $ (\kappa_C - 6\kappa)/2 -\kappa \geq 0.4   \kappa_C  $. 
Recall that $w$ is the matching returned by the matching algorithm during the decoding.

Consider the induced self matching $\tilde{w}$ on $z_?$ from $w$.
There are also  $\geq 0.4 \kappa_C   $   bad matches  in $\tilde{w}$ by the definition of bad matches. 
As $s$ is a $(\Lambda'= 0.2\kappa_C, a)$ synchronization separator sequence, there are only at most $ \Lambda'  $ matches   are undesired.  
The remaining ones, has number at least $ 0.4 \kappa_C - 0.2 \kappa_C  = 0.2\kappa_C$, are desired which means each of them will contribute one to $\cost(\tilde{w})$.
Hence by Lemma \ref{lem:costwlarge}, $\cost(w) = \cost(\tilde{w}) \geq 0.2\kappa_C$.

As a result,  $\obj(w)$ should be at most $ \tau + \kappa - 0.2 \kappa_C = \tau - 0.19\kappa_C$, where $\tau$ is the number of non-zero-symbols in $y$. This is because there are at most $\kappa$ inserted non-zero-symbols contributed to the number of matches and the cost function value is at least $0.2\kappa_C$. 

However, the matching algorithm will return a matching $w$ with maximum $\obj$  by Lemma \ref{lem:DPcorrect}.
Let $w^*$ be the natural ideal matching which matches every non-deleted non-zero-symbols to itself.
Then the target function value of $w$ has to be at least that of $w^*$, which is $|w^*| - \cost(w^*) \geq \tau -\kappa - \kappa = \tau - 2\kappa = \tau-0.02\kappa_C$. The inequality is because there can be at most $\kappa$ non-zero-symbols being deleted and thus $|w^*| \geq \tau - \kappa $; Also each insdel can increase the cost function by at most one, so $\kappa$ insdels make the cost function of $w^*$ to be at most $ \kappa $.

Therefore we reach a contradiction.
Thus $\HD(y, y') \leq \kappa_C$. 
So the decoding can get the correct $x$ by the definition of code $C$.

The efficiency of encoding and decoding follows from Lemma \ref{lem:rtime} and Lemma \ref{lem:syncSepSequenExplicit}.
The parameters directly follow  from the construction.
The linearity follows from the same argument as that of Theorem \ref{thm:linearcoderandconstruct}.
\end{proof}

We state the following theorem  to explicitly describe the error rate and information rate. It immediately follows from Theorem \ref{thm:linearcodedetermconstruct}.
\begin{theorem}

If there is an explicit  linear  code for Hamming distance with block length $n_C$, error rate $\delta_C$, information rate $ \gamma_C $, then there is an explicit   linear code for insdel errors, with block length $   O( n_C/\delta^3_C ) $, insdel error rate $\delta = O(\delta^{4}_C)$, information rate $ O(\delta^{3}_C \cdot \gamma_C)$.
\end{theorem}

We remark that if we pick a very good code $C$, with   error rate $\delta_C$, information rate $ \gamma_C  = 1- \Theta(\delta_C)$.
then if $\delta = \kappa/n$ is the error rate for the constructed  code,
then the information rate for it is  $  c_1\delta^{\frac{3}{4}} (1-c_2 \delta^{\frac{1}{4}}) $ for some constant $c, c_1, c_2>1$.
 
By using an explicit asymptotically good code with linear time encoding and decoding  \cite{spielman96}, we can immediately get the following corollary. The overall decoding time is due to Lemma \ref{lem:rtime}.
\begin{corollary}\label{cor:explicit}
There exists an explicit asymptotically good binary  linear code for edit distance, with linear encoding time and decoding time $O(n^4)$, and a polynomial time pre-computation for the encoding matrix.
\end{corollary}
We remark that the pre-computation is mainly used to compute the synchronization separator sequence.

\section{Explicit Affine Codes}
\label{sec:affine}
In this section, we give an explicit construction of binary affine codes for insdel errors.
We first recall the code construction using synchronization strings \cite{haeupler2017synchronization}.
\begin{theorem}[\cite{haeupler2017synchronization}]
\label{thm:syncstringInsdelCode}
Assume $s \in \Sigma^{n}$ is an $\eta$-synchronization string.
$C $ is an $(n, m, d = 2\epsilon n +1)$ linear code for Hamming distance with efficient decoding.

Then attaching $s$ to every codeword of $C$, symbol by symbol, gives a code which can efficiently correct $\frac{2\epsilon n}{1-\eta}$ insertion/deletion errors.

\end{theorem}

Next we give our construction.
 
\begin{construction}
\label{constr:explicitAffineCode}
Let $s \in \Sigma^{n_0}$ be an $\eta$-synchronization string, where $\eta = 0.01$,  alphabet $\Sigma_s$ has size $ O( \frac{1}{\eta^2} )$ from \cite{CHLSW18}, where each symbol's binary representation has length $l_s$.

Let $C_0 $ be an $(n_0, m_0, d_0 = 2\epsilon n_0 +1)$ $\F_{2^{l_0}}$-linear code for Hamming distance, where   $ l_0  = O(\frac{1}{\epsilon^2}), m_0 = (1-\Theta(\epsilon) ) n_0$. One can use suitable off-the-shelf algebraic-geometric codes for this purpose~\cite{shumetal}.\footnote{The algebraic-geometric construction in fact only requires $l_0=O(\log (1/\epsilon))$ to achieve a relative distance of $\Omega(\epsilon)$ and rate of $1-O(\epsilon)$. However, we will insert buffers between the symbols of this codeword of size $t=O(1/\epsilon)$, and we have to take $l_0 \ge \Omega(t/\epsilon)$ so that the rate loss is only $O(\epsilon)$. If we pick a slightly larger $l_0 = O(\epsilon^{-2} \log (1/\epsilon))$, then we can also take $C_0$ to be an $\F_2$-linear code constructed using expander graphs as in \cite{AEL}. Here $\F_2$-linearity of the code means that the sum of two codewords is also in the code, but the code need not be linear over the extension field $\F_{2^{l_0}}$ (in particular, it need not be closed under multiplication by scalars in $\F_{2^{l_0}}$). One can check that $\F_2$-linearity of $C_0$ suffices for our final binary code to be linear over $\F_2$.}

%Algebraic Geometry code can achieve these parameters)\todo{Does $\F_2$-linearity suffices? If so the ABNNR expander based construction can achieve this with no fancy AG needed.}

Let $l = l_s + l_0$.

Let $t = \Theta(\frac{1}{\epsilon}) $.  

Let string $p = 0 \circ \underbrace{1\ldots 1}_{t+1 \text{ number of } 1's} $. Also call this pattern the boundary string.

\medskip
\noindent The encoding  $C:  \mathbb{F}_2^m \xrightarrow[]{} \mathbb{F}_2^n$ 
 is as follows, where $m = m_0 l_0$, $n =  n_0 l ( 1+ \frac{1}{t} ) + n_0(t+2) $.

\smallskip
\noindent
On input $x \in \mathbb{F}_2^{m}$:
\begin{itemize}

    \item View $x$ as in $(\mathbb{F}_2^{l_0})^{m_0}$ by partitioning every $l_0$ bits as an element in $\mathbb{F}_2^{l_0}$.

    \item Compute $  C_0(x)$, and let $y_i, i\in [n_0]$ be the $l_0$ bits corresponding to the  $i$-th coordinate, i.e. think of the $i$-th coordinate as a degree $\leq  l_0$ polynomial with coefficients in $\mathbb{F}_2$, taking   the binary coefficients  of its monomials to be $y_i$. Let $y = y_1 \circ y_2 \circ \cdots \circ y_{n_0}$
    
    \item Let $y'_i = s_i \circ y_i$ where $s_i$ is the binary representation of the $i$-th symbol of $s$; 
    
    \item After every $t$ bits of $y'_i$, insert a $0$-symbol to attain $y''_i$; 
    
    \item Let $z_i = p \circ y''_i $;

    \item The codeword $z = z_1 \circ \cdots \circ z_{n_0} $. 
    
\end{itemize}

\smallskip\noindent
The decoding is as follows.

On input $\tilde{z}$;

\begin{itemize}
  
    \item Start from the beginning of $\tilde{z}$;
    
    \item Locate every appearance of the boundary $p$;
    
    For the $i$-th appearance of $p$,
    \begin{itemize}
        
        \item Take the bits between the $i$-th boundary and the $i+1$-th boundary  to be the  $i$-th block;  
         
        \item Then for this $i$-th block, eliminate every $t+1$-th symbol (suppose to be the inserted $0$) to get $\tilde{y}'_{i} $;
    
    \end{itemize}
    
    \item View $\tilde{y}'$ as s.t.  each symbol $\tilde{y}'_{i} $ is a symbol which is the concatenation of a symbol from $s$ and an element from $ \mathbb{F}_2^{l_0}$;
    
    \item Use the decoding algorithm from Theorem \ref{thm:syncstringInsdelCode} on  $\tilde{y}'$ to get $x$.

\end{itemize}

\end{construction}

\begin{theorem}[Restate of Theorem \ref{Thm:explicitAffineCode}]
For  any $\epsilon > 0$, there  exists  an  explicit  affine   code  over $\mathbb{F}_2$, with rate  $1- \Theta(\epsilon)$, which can be efficiently decoded from any $ O(\epsilon^3  ) $ fraction of insertions and deletions.

\end{theorem}

\begin{proof}
We show that Construction \ref{constr:explicitAffineCode} gives such a code.

First we prove the correctness of the decoding.
Assume there are $k$ insdel errors. 
Since we insert a $0$-symbol after every $t$ bits of $y_i'$ to attain $y''_i$, there are no appearance of $p$ in $y''_i$.
Consider the adversary's insdels on $z$. 
We claim that  each insdel, in the worst case, can corrupt one of the blocks and at the same time delete or insert a block.  
This is because for each insertion or deletion, it can   corrupt a boundary or the   bits between  it and the next boundary to modify a block. Also at the same time, it may create a new boundary to insert a new block, or may corrupt an existing boundary to  delete a block.  
Note that this will not affect any other block  which has no insdels and is not on the left of corrupted boundary. 
Therefore, for $\tilde{y}'$, it can be viewed as having $2k$ insdels or modifications from $y'$. 
Also recall that $\tilde{y}$ has the structure of the code described in Theorem \ref{thm:syncstringInsdelCode}.
Thus by Theorem \ref{thm:syncstringInsdelCode}, we can correct $\frac{2\epsilon n_0}{1-\eta}  = \Theta(k)$ insdel errors if $\epsilon n_0  = \Theta(k)$.

Next we show the error rate and information rate are as stated.
We take $ \eta = 0.01$, $ l_0 = O(   \epsilon^{-2}) $,  $t= O(  \epsilon^{-1})$. So $l_s = \Theta(\eta^{-2}) = \Theta(1)$.
Recall that 
$$n =  n_0 l (1 +   \frac{1}{t}  ) + n_0 |p|=   n_0 ( l_0 + l_s ) (1 +   \frac{1}{t}  ) + n_0 (t+ 2) .$$
So combining with these parameters we get $k = \Theta(\epsilon n_0) =  O(\epsilon^{3}    n ) $.
The information rate is 
$$ \frac{m_0 l_0}{ n } =   \frac{(1-\Theta(\epsilon)) l_0}{(l_0 + l_s)(1+1/t)+(t+2) } =     \frac{  1-\Theta(\epsilon)  }{ (1+\Theta(\epsilon^2))(1+\Theta(\epsilon)) +\Theta(\epsilon) } \geq 1- \Theta(\epsilon) .$$ 

Now we prove that the code $C$ we constructed is affine.
By definition of affine code, we need to show that the space of the code is a linear subspace $\subseteq \mathbb{F}_2^{n}$, plus a shift $w$ in $\mathbb{F}_2^{n}$.
We let $w $ be the  codeword in $C$, which is the codeword for the zero element of  $ \mathbb{F}_2^{m} $.
Recall that the encoding will first compute $ C_0(x)$ and then partition it into binary bits, i.e. attaining $y(x) $. 
After that the encoding  just inserting some boundaries and  symbols to some specific positions of $y(x)$.
This is doing the same as putting $y$'s coordinates to some specific entries of a vector in $\{0, 1\}^n$ setting other entries to be $0$ and then add $w$ to this vector.
So  we can regard $w$ as the shift, and then to show $C$ is affine, we only need to show that that  $A = \{y(x)  \mid x\in \mathbb{F}_2^m \} \subseteq \{0, 1\}^{n_0 l_0}$ is a linear space.
Since coordinates of $A$ are over $\mathbb{F}_2$, any vector in $A$ times an element in $\mathbb{F}_2$ is still in $A$.
On the other hand, we claim that for two arbitrary vectors $u, v \in A$ corresponding to messages $ x_u, x_v$, it is true that $u \oplus v \in A$.
To see this, recall that $  C_0 $  is a linear code and thus $C_0(x_u) + C_0(x_v) = C_0(x_u + x_v) \in C_0$.
Also for each $i\in [n_0]$   the coordinates addition is done over $\mathbb{F}_2^{l_0}$.
So it is  adding polynomials in $\mathbb{F}_2[x]$ whose degrees are at most $l_0$.
Thus  the binary representation of the $i$-th coordinate of $  C_0(x_u + x_v) = C_0(x_u) + C_0(x_v)$ is
the bit-wise xor over the coordinates' binary representations i.e. $u_i \oplus v_i, i\in [n_0]$, where $u_i, v_i \in \{0, 1\}^{l_0}$ are the binary representation of the  $i$-th coordinates of $C_0(x_u), C_0(x_v) $. 
Therefore $u\oplus v = y(x_u + x_v)$. So $ u\oplus v \in A$.
So $A$ is linear and hence $C$ is affine.
\end{proof}

\section{Explicit Systematic Linear Codes}
\label{sec:systematic-linear}
In this section we construct explicit systematic linear codes for insertions and deletions. Our construction is quite simple, and can be described as follows.

\begin{construction}
\label{con:systematic}
Let $C: \F^m_q \to \F^n_q$ be a $q$-ary linear code for insertions and deletions. The systematic linear code is defined below.

\begin{description}
\item[Encoding:] For any message $x \in \F^m_q$, the encoding function is $\Enc(x)=(x, C(x))$. That is, we simply concatenate $x$ with the codeword $C(x)$.
\item[Decoding:] For any received word $y$, simply remove the first $m$ symbols to get a substring $y'$, and run the decoding function of $C$ on $y'$.
\end{description}
\end{construction}

We have the following lemma.

\begin{lemma}
If $C$ is a code that can correct up to $k$ deletions and insertions, then Construction~\ref{con:systematic} gives a code $\widetilde{C}$ that can also correct up to $k$ deletions and insertions, with codeword length $n+m$.
\end{lemma}

\begin{proof}
Let $\Enc(x)=(x, C(x))$ and $y$ be a string obtained from $\Enc(x)$ by at most $k$ deletions and insertions. We only need to show that $y'$ can also be obtained from $C(x)$ by at most $k$ deletions and insertions. To see this, suppose that there are $r$ insertions and $t$ deletions in the $x$ part of $\Enc(x)$, that changes $x$ into $x'$, where $r+t \leq k$. Thus the total number of insertions and deletions in the $C(x)$ part is at most $k-r-t$. Suppose this changes $C(x)$ into $C'(x)$ and thus $y=(x', C'(x))$. 

Note that after $r$ insertions and $t$ deletions, the length of $x'$ is $m+r-t$. Thus to change $C(x)$ into $y'$, we can first use the same at most $k-r-t$ insertions and deletions in the $C(x)$ part, and then either insert or delete $|r-t|$ symbols. The total number of insertions and deletions needed is at most 
\[k-r-t+|r-t| \leq k. \]
\end{proof}

This gives the following corollary.

\begin{corollary}
Suppose there exists an explicit $q$-ary linear code with rate $R$ that can correct up to $\delta$ fraction of insertions and deletions, then there exists an explicit $q$-ary systematic linear code with rate $R/(1+R)$ that can correct up to $\delta/(1+R)$ fraction of insertions and deletions.
\end{corollary}

Combined with Corollary~\ref{cor:explicit} this gives the following theorem.

\begin{theorem}
For any $n$ there exists an explicit systematic linear code $C_n$ with constant rate which can be efficiently encoded in linear time and decoded from $\Theta(n)$ insertions and deletions in time $O(n^4)$. The generator matrix of $C_n$ can be deterministically computed in polynomial time.
\end{theorem}

\begin{remark}
Our explicit construction of a systematic linear code for insertions and deletions simply concatenates the message with another linear code, and this decreases the rate of the code. As we showed in Section~\ref{sec:bound}, existentially this loss of rate can be avoided. On the other hand, if one considers a weaker notion of systematic code, where one only requires the message to appear as a subsequence of the codeword, then again any linear code can be made systematic in this sense by doing a basis change in the generator matrix. However, to make it systematic in the standard sense may require a permutation of the symbols, and this can potentially change the distance of the code. We also mention that for non linear/affine systematic codes, the constructions in \cite{KZXK18,Haeupler19} based on document exchange protocols can correct up to $\delta n$ fraction of insertions and deletions with rate $1-O(\delta \log^2(1/\delta))$. By our bounds in Section~\ref{sec:bound}, such codes cannot be linear, but it remains an interesting question to see if any affine systematic code can achieve these parameters.
\end{remark}
 
\section{Conclusion and Open Questions}
\label{sec:conclusion}
This paper gives novel existential upper and lower bounds for linear and affine insdel codes and their rate-distance tradeoffs. We also give two explicit constructions of efficient codes for these settings. 

The linear insdel codes described in this paper are the first non-trivial such codes. Their performance guarantees disprove the claim which suggested that  the simple $k+1$-fold repetition code might essentially be the best linear code to correct $k$-deletions. In contrast, we show that there exist much more rate-efficient linear insdel codes (of rate approaching $1/2$) that can   correct a constant fraction of insdel errors.  Indeed, as a first order approximation, the results of this paper suggest that linear codes might only lose a factor of two in the rate for their structural simplicity. Indeed over large alphabets, they achieve rate approaching the half-Singleton bound, which we also show to be tight. 
Furthermore, affine insdel codes, which have a similarly simple structure, break even this barrier and could potentially be as efficient as fully-general insdel codes. For affine codes, we even have an explicit construction of rate approaching $1$ that can efficiently correct a constant fraction of insdel errors. For linear codes, our construction to correct a constant fraction of insdel errors has rate boounded away from $0$, but falls well short of the existential rate $1/2$ threshold.

\medskip There are numerous intriguing questions in the subject that remain open making this work the beginning of a new line of research aimed toward painting a more complete picture of the power and limitations of linear and affine codes for synchronization errors. Concretely, the following are some of the several interesting questions brought to the fore by our work:
\begin{itemize}
%\item Does our construction for binary codes extend to linear codes over other fixed finite fields $\F_q$?
\item Can one find a better distribution of run lengths of $0$'s that we intersperse in our construction that improves the rate of our construction? Could we possibly approach a rate of $1/2$ with this method, or are there any inherent limitations to this approach?

%\item Linear codes for Hamming errors can always be made systematic via a basis change without changing their rate or distance. For linear insdel codes a basis change can drastically alter its edit distance and completely destroy its error correction abilities. What can be said about systematic linear or affine codes? 
%In the context of codes correcting a few constant number of deletions, there have been works giving systematic non-linear codes\todo{Cite Sima et al paper}

%It seems worth checking whether our existential result for random linear codes continues to hold for random systematic linear codes (or random systematic linear codes with a random offset). For sure, none of our code constructions are systematic. 

    \item Can one explicitly construct efficient linear insdel codes over large alphabets which get arbitrarily close to the Half-Singleton bound proven in Theorem~\ref{halfsingletonboundupperandlower}?

\item Can one explicitly construct efficient affine codes over large alphabets that approach the Singleton bound similar to the non-affine insdel codes in \cite{haeupler2017synchronization}? A simple candidate would be one which adds a symbol from a synchronization string after every $\frac{1}{\epsilon}$ symbols of the codewords of a good algebraic geometry code. This affects the rate by only a $(1 - \epsilon)$-factor but whether such few synchronization symbols are sufficient to synchronize a string and reduce any insdel errors to the Hamming setting is unclear.

 \item Is there any separation between affine insdel codes and unrestricted general insdel codes?

   \item Obtain better bounds (even non-constructively) on   asymptotic rate-distance trade-offs in the high-rate regime for linear insdel codes over the binary (or fixed $q$-ary) alphabet.
   
   \item What can be said about the zero-rate regime? That is, what fraction of deletions can be corrected by linear (or affine) codes of rate bounded away from $0$? The random coding argument for $q$-ary codes based on the expected length of the longest common subsequence of two random string (and exponential concentration of this value around the expectation) should probably work for random linear codes since we only need pairwise independence (though we have not checked the details). Can one do much better than the random coding bound, as was done in \cite{BGH17} for general codes?
   
   %Also, what can be achieved with explicit, efficiently decodable, constructions? Can some concatenation with suitable outer code and linear codes of small length at the inner level work?

%    \item How does the coding landscape look like for list-decodable linear or affine insdel codes?
    
 %   \item Can one generalize the alphabet-free Half-Singleton bound of $\frac{1}{2}(1 - \delta)$ to a Half-Plotkin bound of $\frac{1}{2}(1 - \frac{q}{q-1} \delta)$?
\end{itemize}

\bibliographystyle{alpha}
\bibliography{ref}

\appendix

\section{Skipped proofs}
\label{app:skipped-proofs}
\subsection{Proof of Claim~\ref{claim:pairwise-ind}}

Assume without loss of generality that $s_1 < s_2 < \cdots < s_t$ and $r_1 < r_2 < \cdots < r_t$. The simplest case is when one of $x$ and $x'$ is $0^m$. Without loss of generality assume $x=0^m$, then $C=0^n$ while $C'$ is now a vector chosen uniformly from $\F^n_q$ (since $x' \neq x$). Thus it is clear that $\Pr[\forall k \in [t], C_{s_k}=C'_{r_k}] \leq q^{-t}$. Now we assume that $x \neq 0^m$ and $x' \neq 0^m$, and show that $\forall k \leq t-1$, we have

\[\Pr[ C_{s_{k+1}}=C'_{r_{k+1}} | \forall \ell \leq k, C_{s_\ell}=C'_{r_\ell}] \leq q^{-1}.\]

Consider the pair $(s_{k+1}, r_{k+1})$. We have two different cases.

\begin{description}
\item [Case 1:] $s_{k+1} \neq r_{k+1}$. Without loss of generality assume $r_{k+1} > s_{k+1}$. Then we can fix all the columns in $G$ with index in $\{s_1, \cdots, s_{k+1}, r_1, \cdots, r_k\}$. This also fixes $\{C_{s_\ell}, C'_{r_\ell}, \ell \leq k\}$ and $C_{s_{k+1}}$. Note that the column in $G$ with index $r_{k+1}$ is still completely uniform, thus $C'_{r_{k+1}}$ is also uniformly chosen from $\F_q$. Hence we have 

\[\Pr[ C_{s_{k+1}}=C'_{r_{k+1}} | \forall \ell \leq k, C_{s_\ell}=C'_{r_\ell}] \leq q^{-1}.\]

\item [Case 2:] $s_{k+1} = r_{k+1}$. We fix all the columns in $G$ with index in $\{s_1, \cdots, s_k,$ $r_1, \cdots, r_k\}$. This also fixes $\{C_{s_\ell}, C'_{r_\ell}, \ell \leq k\}$. Note that the column in $G$ with index $r_{k+1}$ is still completely uniform, and $C_{s_{k+1}}= \langle x, G_{r_{k+1}}\rangle$, $C'_{r_{k+1}}= \langle x', G_{r_{k+1}}\rangle$. Here $G_{r_{k+1}}$ denotes the column in $G$ with index $r_{k+1}$ and $\langle \rangle$ denotes the inner product over $\F_q$. Thus
\[ C_{s_{k+1}}-C'_{r_{k+1}}=\langle x-x', G_{r_{k+1}}\rangle.\]
Since $x-x' \neq 0^m$, we have
\[\Pr[ C_{s_{k+1}}=C'_{r_{k+1}}] =\Pr[ C_{s_{k+1}}-C'_{r_{k+1}}=0] \leq q^{-1}.  \]
\end{description}

\subsection{Proof of Theorem~\ref{thm:existence-systematic}}

Similar to the proof of Theorem~\ref{thm:existence-linear}, we take a random matrix $G \in \F^{m \times n}_q$ where each entry is chosen independently and uniformly from $\F_q$ and consider the encoding $y =x G$, where $x \in \F^m_q$ is the message. Now, we would like the code to have the following two properties.
\begin{description}
\item[Property 1]: For any two different codewords $C, C'$, $\LCS(C, C') < (1-\delta)n$.
\item[Property 2]: The leftmost $m \times m$ submatrix of $G$ has full rank. 
\end{description}
Suppose $G=[M, V]$ is such a matrix where $M$ is the leftmost $m \times m$ submatrix. Change the encoding function to $y=x M^{-1} G=x [I, M^{-1}V]$, this is a systematic linear code. 

Now for any two different messages $x_1 \neq x_2$, we have $\forall i \in \{1, 2\}$, $x_i [I, M^{-1}V] =(x_i M^{-1}) G$. Since $x_1 \neq x_2$, we must have $(x_1 M^{-1}) \neq (x_2 M^{-1})$ and thus 
\[LCS\left((x_1 M^{-1}) G,  (x_2 M^{-1}) G\right) < (1-\delta)n.\] Therefore the code can still correct up to $\delta n$ insertions and deletions.

Note that property 1 is as before and thus it holds with probability $> 1-q^{2m} 2^{2\ent(\delta)n} q^{(\delta-1)n}$. The probability that property 2 holds is 

\[\prod_{i=1}^m (1-1/q^i) > (1-1/q)(1-\sum_{i=2}^m 1/q^i) > (1-1/q)(1-1/(q^2-q))= 1-(q+1)/q^2.\]

Note that this probability is larger than $1/4$ for all $q \geq 2$. Thus to guarantee the existence of such a matrix $G$ we only need $q^{2m} 2^{2\ent(\delta)n} q^{(\delta-1)n} \leq 1/4$. It then suffices to take  $m \log q=\frac{1-\delta}{2}n \log q-\ent(\delta)n-1$. Thus the information rate is 
\[m/n=(1-\delta)/2-\ent(\delta) / \log q-o(1). \]

\subsection{Proof of Lemma~\ref{lem:largedimensionimplieszeros}}
Suppose a vector $v$ is drawn uniformly at random from $A$. Note that each coordinate of $v$ is either always identical to zero or uniformly random and therefore zero with probability at least $1/q$. The expected number of zeros in $v_1$ as well as in $v_2$ is therefore at least $d/q$. Furthermore two coordinates of $v$ are either pairwise independent or always fixed multiples of each other (and therefore either both or neither zero)---we call two such coordinates dependent. 

Consider the maximal sets of dependent coordinates, i.e., the connected components in the graph on all coordinates of $v$ which connects any two coordinates if they are dependent. Note that an always-zero coordinate is pairwise independent from any other coordinate and therefore forms its own component/maximal dependent set. For every maximal set $S$ of  dependent coordinates the restriction of $A$ onto the coordinates in $S$ is a one (or zero) dimensional subspace. It is therefore possible to find a subspace $A' \subseteq A$ of one smaller dimension in which all coordinates in $S$ are always zero and therefore now each form their own maximal dependent set. Furthermore, the remaining maximal dependent sets $S'$ in $A'$ remain the same as for $A$. We use this projection $t/2$ times on the $t/2$ largest maximal dependent sets of coordinates. This results in a subspace $A'' \subset A$ of dimension $t/2$ in which the largest maximal dependent set is of size at most $\frac{2d}{t/2} = \frac{4d}{t}$.

The number of zeros $Z_1$ in the first $d$ coordinates when drawing a random vector from $A''$ can now be seen as a weighted sum $Z_1 = \sum_S w_S X_S$ of pairwise independent Bernoulli variables $X_S$, one for each maximal dependent set $S$. Here the weight $w_S$ equals the number of coordinates in $S$ that are among the first $d$ coordinates so that $\sum_S w_S = d$. Also $\P[X_S = 1]=\frac{1}{q}$ unless $S$ consists of a single always-zero coordinate in which case $\P[X_S = 1]=1$. Thanks to our construction of $A''$ each weight $w_S$ is at most $\frac{4d}{t}$. Overall we have \[\mathbb{E}[Z_1] = \sum_S w_S \mathbb{E}[X_S] \geq \frac{1}{q}\sum_S w_S = \frac{d}{q}\]
and 
\begin{align*}
\mathrm{Var}[Z_1] = \sum_{S} \mathrm{Var}[w_s X_S] & = \sum_{S} w_S^2 \mathrm{Var}[X_S] \\
& \leq (\max_{S} w_S)(\sum_{S} w_S) \cdot \frac{1}{q} \\
&\leq \frac{4d}{t} \cdot d \cdot\frac{1}{q} \\
& = \frac{4d^2}{qt} \\
& = \left(\frac{2d}{q}\right)^2  \cdot \frac{q}{t}.
\end{align*}
Applying Chebychev's inequality gives
\[\mathrm{Pr}\left[Z_1 \leq \mathbb{E}[Z_1]- 2 \sqrt{\mathrm{Var}[Z_1]} \right] \leq \frac{1}{4} \]

The exact same argument and equation also holds for the number $Z_2$ of zeros in the last $d$ coordinates of a random vector from $A''$. Therefore, the probability to draw a non-zero vector from $A''$ with at least $\mathbb{E}[Z_1]- 2 \sqrt{\mathrm{Var}[Z_1]} \geq \frac{d}{q} - 2 \left(\frac{2d}{q}\right)  \cdot \sqrt{\frac{q}{t}} \ge\frac{d}{q} (1 - 4 \sqrt{\frac{q}{t}})$ zeros in both halves is at least $1 - 2\cdot\frac{1}{4} - q^{-t/2} > 0$. (The $q^{-t/2}$ term is to account for the possibility of sampling the $0$ vector from $A''$.) This implies the existence of such a vector in $A'' \subset A$.

\end{document}